\documentclass[10pt,tightenlines,eqsecnum,floats,aps,amsmath,
amssymb,nofootinbib,prd,showpacs,notitlepage]{revtex4-1}
\usepackage{amsmath,amsthm,latexsym,amssymb,amsfonts}
\usepackage{subfig}
\usepackage[utf8]{inputenc}
\usepackage[T1]{fontenc}
\usepackage{enumerate}
\usepackage{graphicx}
\usepackage{calc}
\usepackage[usenames,dvipsnames]{color}
\usepackage{colordvi}
\usepackage{vmargin}
\usepackage{pstricks}
\usepackage{pst-node}
\usepackage{pst-blur}
\definecolor{czerwony}{HTML}{ec0a25}
\definecolor{zielony}{HTML}{00c000}
\definecolor{niebieski}{HTML}{68619b}
\definecolor{szary}{HTML}{5c5757}
\definecolor{pomaranczowy}{HTML}{8c0000} 
\definecolor{Szary}{HTML}{383535}
\definecolor{Pink}{rgb}{1.,0.75,0.8}

\newcommand{\N}{{\mathbb N}}
\newcommand{\NN}{{\frac{1}{2}\mathbb N}}
\newcommand{\omicron}{{\cal I}}
\newcommand{\Hil}{{\mathcal H}}

\newtheorem{lm}{Lemma}
\newtheorem{thm}[lm]{Theorem}
\newtheorem{df}[lm]{Definition}

\begin{document}

\title{The kernel and the injectivity of the EPRL map} 

\author{Wojciech Kami\'nski${}^1$, Marcin Kisielowski${}^2$, Jerzy Lewandowski${}^2$}

\affiliation{${}^1$Max-Planck-Institut für Gravitationsphysik (Albert-Einstein-Institut),
 Am Mühlenberg 1 D-14476 Golm, Germany\\
${}^2$Instytut Fizyki Teoretycznej, Uniwersytet Warszawski,
ul. Ho{\.z}a 69, 00-681 Warszawa (Warsaw), Polska (Poland)
}

\begin{abstract} \noindent{\bf Abstract\ }
In this paper we prove injectivity of the EPRL map for $|\gamma|<1$, filling the gap of our previous paper. 
\end{abstract}

\maketitle

\section{Introduction}
The Engle-Pereira-Rovelli-Livine (EPRL) map \cite{EPRL} (see also \cite{FK,SFLQG,cEPRL}) is used to define the spin foam  amplitudes between the states of Loop Quantum Gravity \cite{LQGrevsbooks,Baezintro,perez}. The states are labelled by the SU(2) invariants, whereas the gauge group of the EPRL model is Spin(4) or Spin(3,1) depending on the considered spacetime signature. The EPRL map carries the invariants of the tensor products of SU(2) representations into the invariants of the tensor products of Spin(4), or Spin(3,1) representations. Those  LQG states which are not in the domain of or happen to be annihilated by the EPRL map are not given a chance to play a role in the physical Hilbert space. Therefore, it is important to understand which states of LQG are not annihilated.  In the Spin(4) case this issue is particularly subtle, because both the SU(2) representations as well as the Spin(4) representations are labelled by elements of $\frac{1}{2}\N$ , the EPRL map involves rescaling
by constants depending on the Barbero-Immirzi parameter $\gamma$, and the labels (taking values in $\NN$) before and after the map have to sum to an integer. The `injectivity' we prove in the current paper means, that given a necessarily rational value  $\gamma\in \mathbb{Q}$, for every $k_1,...,k_n$ -- n-tuple of elements of $\NN$ -- the EPRL map defined in the space of invariants Inv${\Hil_{k_1}\otimes...\otimes\Hil_{k_n}}$ is injective unless the  target Hilbert space of the corresponding Spin(4) invariants is trivial. The issue of the injectivity of the EPRL map has been raised in \cite{SFLQG}
and \cite{cEPRL}. However, the assumption that ``the  target Hilbert space of the corresponding Spin(4) invariants is not trivial'' was overlooked there. After adding this assumption, the proof presented in \cite{SFLQG} for $\gamma \ge 1$ works without
any additional corrections. Hence, in the current paper we consider only the case of $|\gamma| < 1$. In this case, the theorem formulated in \cite{cEPRL} is true if we additionally assume that the values of $\gamma=\frac{p}{q}$ are such that non of the relatively primary numbers $p$ or $q$ is even. In the current paper we formulate and prove an injectivity theorem valid for every $\gamma\in\mathbb{Q}$, and provide a proof for $|\gamma|<1$.

\setlength{\unitlength}{1cm}

\section{The EPRL map, the missed states, the annihilated states and statement of the result}\label{Sec_inject}
\subsubsection{Definition of the EPRL map}
Let $(j,k,l)\in\NN\times\NN\times\NN$ satisfy triangle inequalities and $j+k+l\in \N$. We denote by $C_{j}^{k l}$ the natural isometric embedding $\Hil_{j}\to \Hil_k \otimes \Hil_l$ and by $C^{j}_{k l}$ the adjoint operator. In the index notation we omit $j, k, l$, e.g. $C^{A_1}_{A_2 A_3}:=(C^{j}_{k l})^{A_1}_{A_2 A_3}$.

Let $k_i \in\NN$, $i\in\{1,\ldots,n\}$. We denote by ${\rm Inv}\left(\Hil_{k_1}\otimes\cdots\otimes\Hil_{k_n}\right)$ the subspace of $\Hil_{k_1}\otimes\cdots\otimes\Hil_{k_n}$ consisting of tensors invariant under the action of $SU(2)$ group.

\begin{df}\label{mapa_EPRL}
Given $\gamma\in\mathbb{Q}$ and $k_i \in\NN$, $i\in\{1,\ldots,n\}$ such that $\forall_i\ \ j^\pm_i:=\frac{|1\pm\gamma|}{2} k_i\in \NN$, the Engle-Pereira-Rovelli-Livine map
\begin{align}
\label{Def}
&\iota_{k_1\ldots k_n}\colon  {\rm Inv}\left(\Hil_{k_1}\otimes\cdots\otimes\Hil_{k_n}\right)\rightarrow {\rm Inv}\left(\Hil_{j^+_1}\otimes\cdots\otimes\Hil_{j^+_n}\right)\otimes {\rm Inv}\left(\Hil_{j^-_1}\otimes\cdots\otimes\Hil_{j^-_n}\right)
\end{align}
is defined as follows \cite{EPRL, SFLQG}:
\begin{align}
&\iota_{k_1\ldots k_n}(\omicron)^{A_1^+\ldots A_n^+ A_1^-\ldots A_n^-}=
\omicron^{A_1\ldots A_n} C_{A_1}^{B_1^+ B_1^-}\cdots C_{A_n}^{B_n^+ B_n^-} {P^+}_{B_1^+\ldots B_n^+}^{A_1^+\ldots A_n^+} {P^-}^{A^-_1 \ldots A_n^-}_{B_1^-\ldots B_n^-}\nonumber,
\end{align}
where $P^+:\Hil_{j^+_1}\otimes\cdots\otimes\Hil_{j^+_n}\to {\rm Inv}\left(\Hil_{j^+_1}\otimes\cdots\otimes\Hil_{j^+_n} \right)$, $P^-:\Hil_{j^-_1}\otimes\cdots\otimes\Hil_{j^-_n}\to {\rm Inv}\left(\Hil_{j^-_1}\otimes\cdots\otimes\Hil_{j^-_n}\right)$  are standing for the orthogonal projections.
\end{df}


\subsubsection{The  missed  states} 
Given a value of the Barbero-Immirzi parameter $\gamma$, the EPRL map is defined on an invariant space Inv$(\Hil_{k_1}\otimes...\otimes\Hil_{kn})$, only if 
the spins $k_1,...,k_n\in\frac{1}{2}\mathbb{N}$ are such that also each $\frac{|1\pm\gamma|}{2}k_1,...,\frac{|1\pm\gamma|}{2}k_n\in \frac{1}{2}\mathbb{N}$.  
That is why we are assuming that $\gamma$ is rational, 
$$ \gamma=\frac{p}{q},$$
where $p,q\in \mathbb{Z}$ and they relatively irrational (the fraction can not be
farther reduced.) If we need an explicit formula for $k\in\frac{1}{2}\mathbb{N}$ such that $\frac{|1\pm\gamma|}{2}k\in \frac{1}{2}\mathbb{N}$, we find two possible cases of $\gamma$ and the corresponding formulas for $k$:
\begin{itemize}
\item[(i)] both $p$ and $q$ odd $\Rightarrow$ $k=qs$ where $s\in\frac{1}{2}\mathbb{N}$, 
\item[(ii)] ($p$ even and $q$ odd) or ($p$ odd and $q$ even)  $\Rightarrow$ $k=2qs$ where $s\in\frac{1}{2}\mathbb{N}$. 
\end{itemize}
Invariants involving even one value of spin $k_i$ which is not that of (i), or, respectively,  (ii)
depending on $\gamma$, are not in the domain of the EPRL, hence they are missed
by the map.


\subsubsection{The annihilated states}
Suppose there is given a space of invariants Inv$(\Hil_{k_1}\otimes...\otimes\Hil_{k_n})$ such that each $k_1,...,k_n$
satisfies (i) or, respectively, (ii) above. 
Suppose also the space  is non-trivial, that is
\begin{align}
k_i\ &\le \sum_{i'\not= i}k_{i'}, \ \ \ \ \ i=1,..,n\label{ineqk}\\
\sum_{i}k_i\ &\in\ \mathbb{N}.\label{sumkint}
\end{align}
The target space of the EPRL map ${\rm Inv}\left(\Hil_{k_1}\otimes\cdots\otimes\Hil_{k_n}\right)\rightarrow {\rm Inv}\left(\Hil_{j^+_1}\otimes\cdots\otimes\Hil_{j^+_n}\right)\otimes {\rm Inv}\left(\Hil_{j^-_1}\otimes\cdots\otimes\Hil_{j^-_n}\right)$ is  nontrivial,
if and only if
 \begin{align}
j^\pm_i\ &\le \sum_{i'\not= i}\ j^\pm_i, \ \ \ \ \ i=1,..,n\label{ineqj}\\
\sum_{i}j^\pm_i\ &\in\ \mathbb{N}.\label{sumjint}
\end{align}
Whereas (\ref{ineqk}) does imply (\ref{ineqj}), the second condition 
\begin{equation}
\sum_{i}j^\pm_i\ =\ \frac{1\pm\gamma}{2}\sum_{i}k_i \in\mathbb{N}
\end{equation}  
is not automatically satisfied for arbitrary $\gamma$.

For example, let 
$$\gamma=\frac{1}{4}, \ \ \ \ k_1,k_2,k_3=4.$$
Certainly the space Inv$(\Hil_4\otimes\Hil_4\otimes\Hil_4)$ is non-empty. 
However, 
$$j^-_1,j^-_2,j^-_3=\frac{3}{2}, \ \ \ j^+_1,j^+_2,j^+_3=\frac{5}{2}$$
and 
$$ {\rm Inv}(\Hil_{\frac{3}{2}}\otimes\Hil_{\frac{3}{2}}\otimes\Hil_{\frac{3}{2}})\otimes
{\rm Inv}(\Hil_{\frac{5}{2}}\otimes\Hil_{\frac{5}{2}}\otimes\Hil_{\frac{5}{2}})\ =\ \{0\}\otimes\{0\}.$$ 
In other words, if $\gamma$ is $0.25$ (close to the ``Warsaw value'' $\gamma= .27...$ \cite{DL,M} usually assumed in the literature), then the EPRL map annihilates the SU(2) invariant corresponding to the spins $k_1=k_2=k_3=4$.     

Generally, for $\gamma$ and $k_1,...,k_n$ of the case (i) in the previous subsection, actually  (\ref{sumkint}) does imply (\ref{sumjint}). 
In the case (ii), on the other hand, there is a set of non-trivial subspaces
Inv$(\Hil_{k_1}\otimes...\otimes\Hil_{k_n})$ which are annihilated by the EPRL map
for the target ${\rm Inv}\left(\Hil_{j^+_1}\otimes\cdots\otimes\Hil_{j^+_n}\right)\otimes {\rm Inv}\left(\Hil_{j^-_1}\otimes\cdots\otimes\Hil_{j^-_n}\right)$  
is just the trivial space.                           
 
The theorem formulated below exactly states, that the EPRL map does not annihilate more states, then those characterised above.


\subsubsection{The injectivity theorem} 
\begin{thm}
\label{inj-EPRL}
Assume $\gamma\in\mathbb{Q}$. For $k_i \in\NN$, $i\in\{1,\ldots,n\}$ such that: 
\begin{itemize}
\item $\forall_i\ \ j^\pm_i:=\frac{1\pm\gamma}{2} k_i\in \NN$,
\item $\sum_{i=1}^n j^+_i \in \N$
\end{itemize}
the EPRL map defined above (def. \ref{mapa_EPRL}) is injective.
\end{thm}

Note that when ${\rm Inv}\left(\Hil_{k_1}\otimes\cdots\otimes\Hil_{k_n}\right)$ is trivial, injectivity trivially holds.   

In this article we use the following definition:
\begin{df}
A sequence $(k_1,k_2,\ldots,k_n)$, $k_i\in \NN$ is \textbf{admissible} iff the space ${\rm Inv}\left(\Hil_{k_1}\otimes\cdots\otimes\Hil_{k_n}\right)$ is nontrivial. This is equivalent to the conditions
\begin{equation}
 \forall i\ k_i\leq \sum_{j\not=i} k_j,\quad {\rm and} \quad \sum_i k_i\in \mathbb{N}\ .
\end{equation}
\end{df}

For $\gamma \ge 1$ the proof of the theorem \ref{inj-EPRL} is presented in \cite{SFLQG}. Here we present the proof for $|\gamma|<1$. In order to make the presentation clear, we divide it into sections. The main result is an inductive hypothesis stated and proved in section \ref{inductive_hyp}. The injectivity of EPRL map follows from that result. In the preceding section \ref{simpl_proof} we present proof of theorem restricted to certain intertwiners which we call tree-irreducible. We use it in the proof of the main result.

\section{Proof of the theorem in simplified case}\label{simpl_proof}
\subsection{Tree-irreducible case of inductive hypothesis}
\label{inj-simplified}
In the tree-irreducible case we restrict to intertwiners which we call tree-irreducible. We say that $\omicron\in {\rm Inv}\left(\Hil_{k_1}\otimes\cdots\otimes\Hil_{k_n}\right)$ is \textbf{tree-irreducible}, if for all $l\in\{1,\ldots,n-1\}$ the orthogonal projection $P^{l}:{\rm Inv}\left(\Hil_{k_1}\otimes\cdots\otimes\Hil_{k_n}\right)\to {\rm Inv}\left(\Hil_{k_1}\otimes\cdots\otimes\Hil_{k_l}\right)\otimes\left(\Hil_{k_{l+1}}\otimes\ldots\otimes\Hil_{k_n}\right)$ annihilates $\omicron$.

The inductive proof we present needs an extended notion of the EPRL map. This will be the map $\iota$ analogous to EPRL map but defined under a bit different conditions: 

\begin{quote}
{\bf Con $n$:} Sequences $(k_1,\ldots, k_n)$ and $(j^\pm_1,\ldots, j^\pm_n)$, where $k_i,j_i^\pm \in \NN$, are such that
\begin{itemize}
 \item $(k_1,\ldots, k_n)$ is admissible,
\item $k_i\not=0$ for all $i>1$,
\item $j^+_1+j^-_1=k_1$,
\item $j^\pm_i=\frac{1\pm\gamma}{2} k_i$ for $i\not=1$ and
\begin{equation}
 \frac{1+\gamma}{2}k_1-\frac{1}{2}\leq j^+_1\leq \frac{1+\gamma}{2}k_1+\frac{1}{2} 
\end{equation}
\item $j^\pm_1+ \ldots + j^\pm_n \in \N$
\item (ordering) $\exists i\geq 1\colon \begin{array}{ll}
                                         j_l^+\in \mathbb{N} & l\leq i\\
					 j_l^+\in \mathbb{N}+\frac{1}{2} & l>i
                                        \end{array}$.
\end{itemize}
\end{quote}
Few remarks are worth to mention:
\begin{itemize}
\item We would like to emphasize that although EPRL map usually do not satisfy those conditions, it can be easily replaced by an equivalent map that satisfies {\bf Con n}.

First of all we can assume that in the EPRL map $k_i\not=0$ for $i\geq 1$. Secondly, we can permute $k_i$ in such a way that 
\begin{equation}
 \exists i\colon\ j^+_l\in \N,\ {\rm for}\ l\leq i,\ \ \ j^+_l\in \N+\frac{1}{2},\ {\rm for}\ l\leq i.\ \ 
\end{equation}
These are exactly conditions of {\bf Con n}.

\item From the definition above follows that $j^+_i+j^-_i=k_i$ for all $i=1,\ldots,n$.
\item It follows also that
\begin{equation*}
 \frac{1-\gamma}{2}k_1-\frac{1}{2}\leq j^-_1\leq \frac{1-\gamma}{2}k_1+\frac{1}{2}. 
\end{equation*}
\item From conditions {\bf Con $n$} follows that $(j^\pm_1, \ldots, j^\pm_n)$ satisfy admissibility conditions -- this will be proved in lemma \ref{inequality}.
\end{itemize}

\begin{lm}
\label{inequality}
Let $(k_1,\ldots, k_n)$ and $(j^\pm_1,\ldots, j^\pm_n)$ be elements of $\NN$, such that: $(k_1,\ldots, k_n)$ is admissible, $j^\pm_i=\frac{1\pm\gamma}{2} k_i$ for $i\not=1$,
$\frac{1+\gamma}{2}k_1+\frac{1}{2}\geq j^+_1\geq \frac{1+\gamma}{2}k_1-\frac{1}{2},
$ $j^+_1+j^-_1=k_1$, $j^\pm_1+\ldots + j^\pm_n\in \N$, then $(j^\pm_1,\ldots, j^\pm_n)$ satisfy admissibility conditions.
\end{lm}
\begin{proof}
From the definition of $j^\pm_i$ and from the fact that $(k_1,\ldots, k_n)$ are admissible, we know that \[j^\pm_1\leq \frac{1\pm\gamma}{2} k_1 +\frac{1}{2}\leq \frac{1\pm\gamma}{2} (k_2+\ldots+k_n) +\frac{1}{2}=j^{\pm}_2+\ldots+j^{\pm}_n + \frac{1}{2}.\] 
We have  $j^\pm_1+\ldots + j^\pm_n\in \N$, so $j^\pm_1<j^{\pm}_2+\ldots+j^{\pm}_n + \frac{1}{2}$. As a result $j^\pm_1\leq j^{\pm}_2+\ldots+j^{\pm}_n$. This is one of the desired inequalities.

Similarly for $i\neq 1$ we have:
\[
 j^{\pm}_i = \frac{1\pm \gamma}{2} k_i \leq \frac{1\pm \gamma}{2} k_1 + \sum_{l>1, l\neq i}\frac{1\pm \gamma}{2} k_l \leq \sum_{l\neq i} j^{\pm}_l + \frac{1}{2}
\]
As in previous case $j^\pm_1+\ldots + j^\pm_n\in \N$ implies $j^{\pm}_i<\sum_{l\neq i} j^{\pm}_l + \frac{1}{2}$ and finally $j^{\pm}_i\leq \sum_{l\neq i} j^{\pm}_l$. This finishes proof of this lemma.

\end{proof}

We will base the proof of theorem \ref{inj-EPRL}, in the case $\omicron$ is tree-irreducible, on the following inductive hypothesis ($n\in \N_+$, $n\geq 3$):
\begin{quote}
{\bf Hyp $n$:} Suppose that $(k_1,\ldots, k_n)$, $(j^\pm_1,\ldots, j^\pm_n)$ satisfy condition {\bf Con $n$} and that $\omicron\in 
{\rm Inv}\left(\Hil_{k_1}\otimes\cdots\otimes\Hil_{k_n}\right)$ is tree-irreducible. Then, there 
exists 
$$\phi\in {\rm Inv}\left(\Hil_{j^+_1}\otimes\cdots\otimes\Hil_{j^+_n}\right)\otimes {\rm Inv}\left(\Hil_{j^-_1}\otimes\cdots\otimes\Hil_{j^-_n}\right)$$ such that
$\langle \phi, \iota_{k_1\ldots k_n}(\omicron)\rangle\not=0$.
\end{quote}
This in fact proves injectivity, when restricting to tree-irreducible intertwiners. Note that $\langle \phi, \iota_{k_1\ldots k_n}(\omicron)\rangle = \langle \phi, \iota'_{k_1\ldots k_n}(\omicron)\rangle$, where $\iota'_{k_1\ldots k_n}$ is defined without projections onto invariants of $Spin(4)$, i.e.
\begin{align*}
&\iota'_{k_1\ldots k_n}\colon  {\rm Inv}\left(\Hil_{k_1}\otimes\cdots\otimes\Hil_{k_n}\right)\rightarrow \left(\Hil_{j^+_1}\otimes\cdots\otimes\Hil_{j^+_n}\right)\otimes \left(\Hil_{j^-_1}\otimes\cdots\otimes\Hil_{j^-_n}\right)\\
&\iota'_{k_1\ldots k_n}(\omicron)^{A_1^+\ldots A_n^+ A_1^-\ldots A_n^-}=
\omicron^{A_1\ldots A_n} C_{A_1}^{A_1^+ A_1^-}\cdots C_{A_n}^{A_n^+ A_n^-}
\end{align*}
As a result,  it is enough to find $\phi$, such that $\langle \phi, \iota'_{k_1\ldots k_n}(\omicron)\rangle\not=0$.


\subsection{Proof of tree-irreducible case of inductive hypothesis} \label{simpl_ind_step}
We present in this section the proof in this tree-irreducible case. To make the presentation more transparent, we move some parts to sections \ref{choice_j_alpha} and \ref{chi_not_0}.

 Assume that $n > 3$ and we have proved {\bf Hyp $n-1$}. Let $(k_1,\ldots, k_n)$ and $(j^\pm_1,\ldots, j^\pm_n)$ satisfy {\bf Con $n$} and $\omicron \in {\rm Inv}\left(\Hil_{k_1}\otimes\cdots\otimes\Hil_{k_n}\right)$ is tree-irreducible. We may write the invariant in the following way:
\begin{equation}
 \label{inv_dec}
 \omicron^{A_1 A_2 \ldots A_n}=\sum_{k_{\alpha} \in J} C^{A_1 A_2}_{A_\alpha} (\omicron^{k_{\alpha}})^{A_{\alpha} A_3 \ldots A_n},
\end{equation}
where $J:=\{k_{\alpha}\in \NN: \omicron^{k_{\alpha}}\not\equiv 0 \}$.

\begin{enumerate}
 \item Define $k'_{\alpha}$ to be the minimal element in $J$. Note that, if $n>2$, then $k'_{\alpha}\not=0$, because $\omicron$ is tree-irreducible.
 \item Find $j^{+}_{\alpha}$ (determined by $k'_{\alpha}$) using the procedure defined in section \ref{choice_j_alpha}. This procedure uses the fact that $\omicron$ is tree-irreducible.

 As a result we obtain $j^{+}_{\alpha}\in \NN$, such that:
\[
 \frac{1+\gamma}{2} k'_\alpha - \frac{1}{2} \leq j^{+}_{\alpha}\leq \frac{1+\gamma}{2} k'_\alpha + \frac{1}{2},
\]
$(j^+_{\alpha},j_1^+,j_2^+)$ and  $(j^-_{\alpha},j_1^-,j_2^-)$ are admissible ($j^-_{\alpha}:=k'_{\alpha}-j^+_{\alpha}$).

Note that $j^{\pm}_{\alpha}+j^{\pm}_{3}+\ldots+ j^{\pm}_n \in \N$. It follows from the fact that $j^{\pm}_1,\ldots, j^{\pm}_n \in \frac{1}{2}\N$ (i.e. from {\bf Con $n$}) and the fact that $j^{\pm}_{\alpha}+j_1^{\pm}+j_2^{\pm}\in \N$. 

Let us also notice, that $j^+_\alpha\in \N+\frac{1}{2}$ only if exactly one of $j^+_1$ or $j_2^+$ belongs to $\N+\frac{1}{2}$. Then from the ordering condition, only $j^+_2\in \N+\frac{1}{2}$ and so $j^+_\alpha,j^+_3,\ldots, j^+_n\in \N+\frac{1}{2}$. Ordering condition is thus satisfied also for $(k_\alpha,k_3,\ldots, k_n)$.
 \item Considerations above show that $(k_{\alpha},k_3,\ldots,k_n)$ and $(j_{\alpha}^{\pm},j_3^{\pm},\ldots,j_n^{\pm})$ satisfy {\bf Con $n-1$}. Moreover $\omicron^{k'_{\alpha}}$ is tree-irreducible, because $\omicron$ is.

 From {\bf Hyp $n-1$} follows that for $\omicron^{k'_{\alpha}}$ there exists 
\[\phi^{k'_{\alpha}} \in {\rm Inv}\left(\Hil_{j^+_{\alpha}}\otimes\cdots\otimes\Hil_{j^+_n}\right)\otimes {\rm Inv}\left(\Hil_{j^-_{\alpha}}\otimes\cdots\otimes\Hil_{j^-_n}\right)\] such that
$\langle \phi^{k'_{\alpha}}, \iota'_{k_\alpha k_3\ldots k_n}(\omicron^{k'_{\alpha}})\rangle\not=0$.
\item Having defined $\phi^{k'_{\alpha}}$, we construct $\phi$:
\[
 \phi^{A_1^+\ldots A_n^+, A_1^{-}\ldots A_n^{-}}:=C^{A_1^+ A_2^+}_{A_{\alpha}^+} C^{A_1^- A_2^-}_{A_{\alpha}^-} (\phi^{k'_{\alpha}})^{A_{\alpha} A_3^+\ldots A_n^+, A_{\alpha}^- A_3^{-}\ldots A_n^{-}}
\]
\item The $\phi$ constructed in previous point is the $\phi$ we are looking for, i.e.
$\langle \phi, \iota'_{k_1\ldots k_n}(\omicron)\rangle\not=0$. In this point we show it.

First, using equation \eqref{inv_dec} we write $\langle \phi, \iota'_{k_1\ldots k_n}(\omicron)\rangle$ as a sum:
\begin{equation}\label{sum_1}
 \langle \phi, \iota'_{k_1\ldots k_n}(\omicron)\rangle = \sum_{k_{\alpha}} \langle \phi, \iota'_{k_1\ldots k_n}(C_{k_\alpha}^{k_1 k_2} \circ \omicron^{k_{\alpha}})\rangle,
\end{equation}
where ${(C_{k_\alpha}^{k_1 k_2} \circ \omicron^{k_{\alpha}})}^{A_1 A_2 \ldots A_n}:= C^{A_1 A_2}_{A_{\alpha}} (\omicron^{k_{\alpha}})^{A_{\alpha} A_3 \ldots A_n}$.

From the definition of $k'_{\alpha}$ in point 1 follows that the sum is actually over $k_{\alpha}\geq k'_{\alpha}$:
\begin{equation}
\label{sum}
  \langle \phi, \iota'_{k_1\ldots k_n}(\omicron)\rangle = \sum_{k_{\alpha}\geq k'_{\alpha}} \langle \phi, \iota'_{k_1\ldots k_n}(C_{k_\alpha}^{k_1 k_2} \circ \omicron^{k_{\alpha}})\rangle .
\end{equation}

Let us compute each term $\langle \phi, \iota'_{k_1\ldots k_n}(C_{k_\alpha}^{k_1 k_2}\circ \omicron^{k_{\alpha}})\rangle$ (such term is schematically illustrated on picture \ref{fig:simplified_1}):
\begin{figure}[hbt!]
  \centering
\subfloat[$\langle \phi, {\color{czerwony}\iota'_{k_1\ldots k_n}}(C_{k_\alpha}^{k_1 k_2}\circ {\color{zielony} \omicron^{k_{\alpha}}})\rangle=\langle C^{j_1^+ j_2^+}_{j_{\alpha}^+}\circ C^{j_1^- j_2^-}_{j_{\alpha}^-}{\color{niebieski} \phi^{k'_{\alpha}}},{\color{czerwony} \iota'_{k_1\ldots k_n}}(C_{k_\alpha}^{k_1 k_2}\circ{\color{zielony} \omicron^{k_{\alpha}}})\rangle$
]{\label{fig:simplified_1}\includegraphics[width=0.5\textwidth, trim=50px 100px 0px 60px]{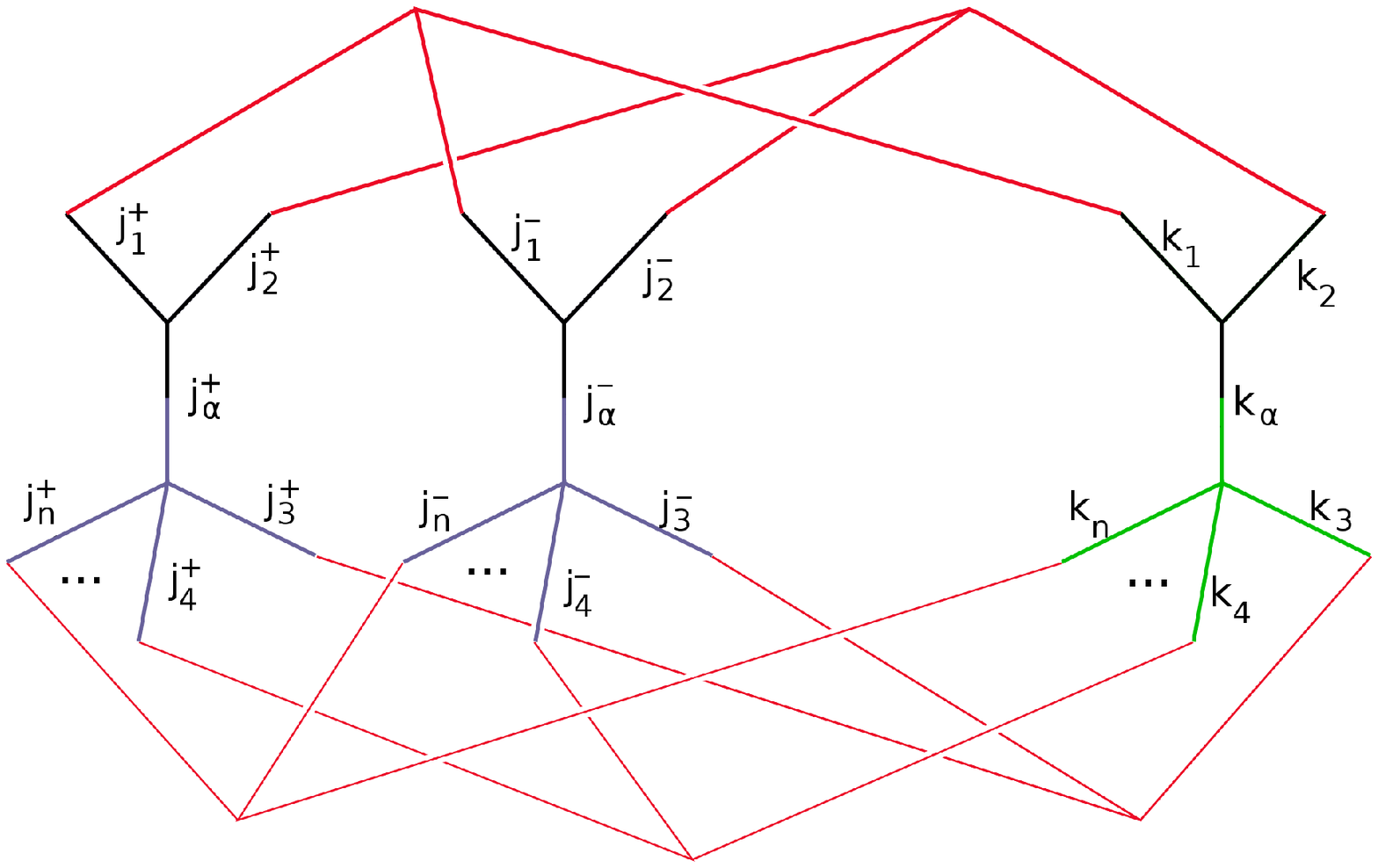} }
\subfloat[The only non-trivial term in the sum \eqref{sum_1} is $\chi \langle {\color{niebieski} \phi^{k'_{\alpha}}},{\color{czerwony} \iota'_{k_\alpha k_3\ldots k_n}}({\color{zielony}\omicron^{k'_{\alpha}}})\rangle$.]{\label{fig:simplified_2}
\includegraphics[width=0.50\textwidth, trim=0px 50px 0px 60px]{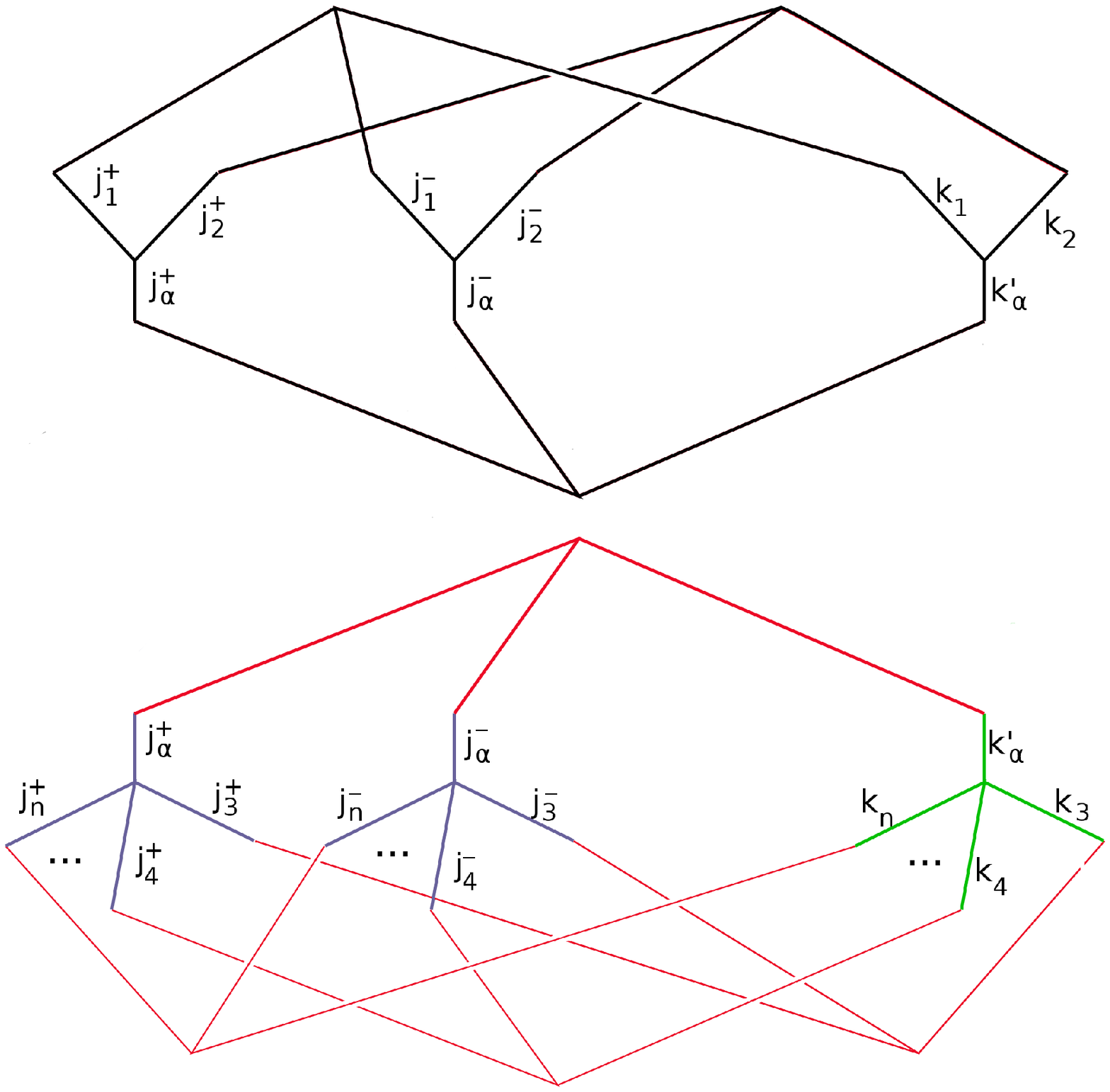}}
  \caption{$\langle \phi, \iota'_{k_1\ldots k_n}(\omicron)\rangle$ equals the sum of terms depicted on figure \ref{fig:simplified_1}. The only non-trivial term is depicted on figure \ref{fig:simplified_2}. Its non-triviality follows from {\bf Hyp $n-1$}} and lemma \ref{nonzero}.
  \label{fig:simplified}
\end{figure}
\begin{align*}
\langle \phi, \iota'_{k_1\ldots k_n}(& C_{k_\alpha}^{k_1 k_2}\circ \omicron^{k_{\alpha}})\rangle=\\=
(\phi^{k'_\alpha})^\dagger_{A^+_\alpha A_3^+\ldots A_n^+ A^-_\alpha A^-_3 \ldots A_n^-} C_{A^+_1 A^+_2}^{A^+_\alpha} & C_{A^-_1 A^-_2}^{A^-_\alpha} C_{A_1}^{A^+_1 A^-_1}\cdots C_{A_n}^{A^+_n A^-_n} C_{A_\alpha}^{A_1 A_2} (\omicron^{k_\alpha})^{A_\alpha \ldots A_n} \ \ =\\
= C_{A^+_1 A^+_2}^{A^+_\alpha} C_{A^-_1 A^-_2}^{A^-_\alpha} C_{A_1}^{A^+_1 A^-_1} C_{A_2}^{A^+_2 A^-_2} C_{A_\alpha}^{A_1 A_2} & (\phi^{k'_\alpha})^\dagger_{A^+_\alpha A_3^+\ldots A_n^+ A^-_\alpha A^-_3 \ldots A_n^-} C_{A_3}^{A^+_3 A^-_3}\cdots C_{A_n}^{A^+_n A^-_n} (\omicron^{k_\alpha})^{A_\alpha A_3 \ldots A_n}
\end{align*}
We have 
\begin{equation}\label{clebsch}
 C_{A^+_1 A^+_2}^{A^+_\alpha} C_{A^-_1 A^-_2}^{A^-_\alpha} C_{A_1}^{A^+_1 A^-_1} C_{A_2}^{A^+_2 A^-_2} C_{A_\alpha}^{A_1 A_2}=
\left\{\begin{array}{ll}
0 & k_\alpha>j^+_\alpha+j^-_\alpha\\
\chi C_{A_\alpha}^{A^+_\alpha A^-_\alpha}, & k_\beta=j^+_\alpha+j^-_\alpha.
\end{array}\right.
\end{equation}
The first equality is obvious because there exists no intertwiner if $k_\alpha>j^+_\alpha+j^-_\alpha$ (let us remind that $j^+_\alpha+j^-_\alpha=k'_{\alpha}$). The second equality is also obvious because for $k_\alpha=j^+_\alpha+j^-_\alpha$, the space ${\rm Inv}\left( \Hil_{k_{\alpha}}\otimes \Hil_{k_2}\otimes \Hil_{k_3} \right)$ is one-dimensional. The nontrivial statement is that $\chi\not= 0$. The non-triviality of $\chi$ is assured by lemma \ref{nonzero} which was proved in our previous article \cite{cEPRL}.

Summarizing, for some $\chi\in\mathbb{C}\backslash\{0\}$, we have:
\[
 \langle \phi, \iota'_{k_1\ldots k_n} (C^{k_1 k_2}_{k_\alpha}\circ \omicron^{k_{\alpha}}) \rangle=
\left\{\begin{array}{ll}
        0, & k_\alpha>k'_\alpha\\
\chi \langle \phi^{k'_{\alpha}}, \iota'_{k_\alpha k_3\ldots k_n} (\omicron^{k'_{\alpha}}) \rangle, & k_\alpha=k'_\alpha\\ 
*, & k_\alpha<k'_\alpha
       \end{array}\right.
\]
As a result all but one term in the sum \eqref{sum} are equal zero and:
\[
 \langle \phi, \iota'_{k_1\ldots k_n}(\omicron)\rangle=\chi \langle \phi^{k'_{\alpha}}, \iota'_{k_\alpha k_3\ldots k_n} (\omicron^{k'_{\alpha}})\rangle \not=0
\]
\end{enumerate}
We obtained that for $n> 3$, {\bf Hyp $n$} follows from {\bf Hyp $n-1$}. In order to finish the inductive proof, it remains to check that {\bf Hyp $3$} is true. In this case sequences $(k_1,k_2,k_3)$ and $(j^\pm_1,j^\pm_2,j^\pm_3)$ are admissible and invariant spaces are one dimensional. {\bf Hyp 3} follows now from lemma \ref{oldlemma}.

This proof of first inductive step is valid in general case, because for $n=3$ all invariants are tree-irreducible.

\subsection{The choice of $j_{\alpha}^+$}\label{choice_j_alpha}
In this section we discuss the procedure of choosing $j_{\alpha}^+$. It is depicted on the diagram below and it is justified by three lemmas \ref{jeden}, \ref{dwa}, \ref{trzy}. Note that $k_1\not=0$ and $k_\alpha\not=0$ (on every step of inductive procedure), because $\omicron$ is tree-irreducible. It is reflected in these lemmas by the condition, that $j\not=0$ and $l\not=0$.
\begin{center}
\vspace{0.5cm}
\small
\psframebox[linearc=0.5,cornersize=absolute,framesep=10pt]{%
  \psset{shadowcolor=black!70,blur=true}%
  \begin{psmatrix}[rowsep=0.4,colsep=0.5]
    \psframebox[fillstyle=solid,fillcolor=Pink,shadow=true]{$\frac{1+\gamma}{2} k_1 -\frac{1}{2}<j^+_1<\frac{1+\gamma}{2} k_1 +\frac{1}{2}$ ?} &
       \psframebox[shadow=true]{\begin{tabular}{c} Take one of the $j^{+}_{\alpha}$ satisfying\\ $\frac{1+\gamma}{2} k'_{\alpha} -\frac{1}{2}\leq j_{\alpha}^+\leq \frac{1+\gamma}{2} k'_{\alpha} +\frac{1}{2}$\\ and such that $j^+_\alpha+j_1^++j_2^+\in \N$  \end{tabular}
} \\
    \psframebox[fillstyle=solid,fillcolor=Pink,shadow=true]{\begin{tabular}{c} Does there exist $j^{+}_{\alpha}\in \NN$ satisfying\\ $\frac{1+\gamma}{2} k'_{\alpha} -\frac{1}{2}< j^+_{\alpha}< \frac{1+\gamma}{2} k'_{\alpha} +\frac{1}{2}$\\ and such that $j^+_\alpha+j_1^++j_2^+\in \N$ ? \end{tabular}} &
       \psframebox[shadow=true]{ Take this $j^{+}_{\alpha}$.
} \\
    ~\\
    \psframebox[shadow=true,fillstyle=solid,fillcolor=Pink]{$k_2+k'_{\alpha}>k_1$ and $k_2+k_1>k'_{\alpha}$ ?} & \psframebox[shadow=true]{\begin{tabular}{c} In this case $j_1^+=\frac{1+\gamma}{2} k_1 +\frac{1}{2}$ \\or $j_1^+=\frac{1+\gamma}{2} k_1 -\frac{1}{2}$.\\ Take $j_{\alpha}^+=\frac{1+\gamma}{2} k'_{\alpha} -\frac{1}{2}$ \\or $j_{\alpha}^+=\frac{1+\gamma}{2} k'_{\alpha} +\frac{1}{2}$ respectively.\end{tabular}} \\
    \psframebox[shadow=true]{\begin{tabular}{c} In this case $j_1^+=\frac{1+\gamma}{2} k_1 +\frac{1}{2}$ \\ or $j_1^+=\frac{1+\gamma}{2} k_1 -\frac{1}{2}$.\\ Take $j_{\alpha}^+=\frac{1+\gamma}{2} k'_{\alpha} +\frac{1}{2}$ \\ or $j_{\alpha}^+=\frac{1+\gamma}{2} k'_{\alpha} -\frac{1}{2}$ respectively.\end{tabular}}
    \ncline{->}{1,1}{2,1}<{\textcolor{red}{No}}
    \ncline{->}{2,1}{4,1}<{\textcolor{red}{No}}
    \ncline{->}{2,1}{2,2}^{\textcolor{red}{\hspace{50px} Yes}}
    \ncline{->}{4,1}{4,2}^{\textcolor{red}{Yes}}
    \ncline{->}{4,1}{5,1}<{\textcolor{red}{No}}
    \ncline{->}{1,1}{1,2}^{\textcolor{red}{Yes}}
    \ncline{->}{1,2}{1,3}
    \ncbar[angleA=-90,armB=0,nodesepB=0.25]{->}{3,3}{4,1}
    \end{psmatrix}%
}
\end{center}

We define $j^-_{\alpha}:= k'_{\alpha}-j^+_{\alpha}$. In each case in the diagram above lemmas \ref{jeden}, \ref{dwa}, \ref{trzy} show that $(j^+_{\alpha},j_1^+,j_2^+)$ and  $(j^-_{\alpha},j_1^-,j_2^-)$ are admissible. First lemma is used in first and second case depicted in the diagram (in those cases we use lemma \ref{jeden} with $j= k_1, k= k_2 ,l= k_{\alpha}, k^{\pm}=j^{\pm}_1, j^{\pm}=j^{\pm}_2, l^{\pm}=j^{\pm}_{\alpha}$ and $j=k_{\alpha}, k= k_2, l= k_1, k^{\pm}=j^{\pm}_{\alpha}, j^{\pm}=j^{\pm}_2, l^{\pm}=j^{\pm}_1$ respectively). Second and third lemma is used in the last step. We prove now those lemmas.
\begin{lm}
 Let $(j,k,l)$ be admissible and $j\not=0,l\not= 0$. If $j^{+},k^{+},l^{+}$ are elements of $\NN$ satisfying:
$\frac{1+\gamma}{2}j-\frac{1}{2}<j^+< \frac{1+\gamma}{2}j+\frac{1}{2}$, $k^+=\frac{1+\gamma}{2}k$, $\frac{1+\gamma}{2}l-\frac{1}{2}\leq l^+\leq \frac{1+\gamma}{2}l+\frac{1}{2}$ and $j^{+}+k^{+}+l^{+}\in \N$, then $(j^{+},k^{+},l^{+})$ and $(j-j^+,k-k^+,l-l^+)$ are admissible.
\label{jeden}
\end{lm}
\begin{proof}
We denote $j^{-}:=j-j^+,k^{-}:=k-k^+,l^{-}:=l-l^+$.
\begin{enumerate}
 \item Notice that $j^{-},k^{-},l^{-}$ satisfy $\frac{1-\gamma}{2}j-\frac{1}{2}<j^-< \frac{1-\gamma}{2}j+\frac{1}{2}$, $k^-=\frac{1-\gamma}{2}k$, $\frac{1-\gamma}{2}l-\frac{1}{2}\leq l^-\leq \frac{1-\gamma}{2}l+\frac{1}{2}$ and $j^{-}+k^{-}+l^{-}\in \N$. It is a direct check. Inequalities $\frac{1+\gamma}{2}j-\frac{1}{2}<j^+< \frac{1+\gamma}{2}j+\frac{1}{2}$ imply, that
\[
 \frac{1+\gamma}{2}j-\frac{1}{2}<j-j^-< \frac{1+\gamma}{2}j+\frac{1}{2}.
\]
As a result
\[
 \frac{-1+\gamma}{2}j-\frac{1}{2}<-j^-< \frac{-1+\gamma}{2}j+\frac{1}{2}
\]
and
\[
 \frac{1-\gamma}{2}j-\frac{1}{2}<j^-< \frac{1-\gamma}{2}j+\frac{1}{2}
\]

The same with $\frac{1-\gamma}{2}l-\frac{1}{2}\leq l^-\leq \frac{1-\gamma}{2}l+\frac{1}{2}$ and $k^-=\frac{1-\gamma}{2}k$ is obvious. Finally $j^{-}+k^{-}+l^{-}\in \N$ follows from the fact that $j^{+}+k^{+}+l^{+}\in \N$ and $j+k+l\in \N$.

\item Note  also that $j^{-}\geq 0, k^{-}\geq 0, l^{-}\geq 0$: $\frac{1-\gamma}{2}j-\frac{1}{2}<j^-$, so $-\frac{1}{2}<j^-$; similarly $\frac{1-\gamma}{2}l-\frac{1}{2}\leq l^-$ implies $-\frac{1}{2}<l^-$, because $l\not=0$ and $|\gamma|<1$; $k^{\pm}\geq 0$ is straightforward.

\item We check now triangle inequalities.

\[
 j^++k^+>\frac{1+\gamma}{2} j -\frac{1}{2} + \frac{1+\gamma}{2} k = \frac{1+\gamma}{2}(j+k) -\frac{1}{2}\geq \frac{1+\gamma}{2} l -\frac{1}{2} \geq l^+ -1.
\]
It follows that
\[
 j^++k^+ - l^+> -1.
\]
However $j^++k^+ + l^+ \in \N$, so $j^++k^+ - l^+\in \mathbb{Z}$. As a result
\[
 j^++k^+ - l^+ \geq 0.
\]

Similarly,
\[
 k^++l^+ \geq \frac{1+\gamma}{2} k +\frac{1+\gamma}{2} l -\frac{1}{2} = \frac{1+\gamma}{2}(k+l) -\frac{1}{2}\geq \frac{1+\gamma}{2} j -\frac{1}{2} > j^+ -1.
\]
We obtain $k^++l^+ - j^+\geq 0$.

We have also
\[
 l^++j^+ > \frac{1+\gamma}{2} l +\frac{1+\gamma}{2} j -1 = \frac{1+\gamma}{2}(j+l) -1 \geq \frac{1+\gamma}{2} k -1 = k^+ -1.
\]
Finally $l^++j^+ - k^+\geq 0$. This proves that $(j^+,k^+,l^+)$ is admissible. The proof for  $(j^-,k^-,l^-)$ is the same.
\end{enumerate}
\end{proof}

\begin{lm}
  Let $(j,k,l)$ be admissible and $j\not=0,l\not= 0$. If $j^{+},k^{+},l^{+}$ are elements of $\NN$ satisfying:
$j^+ = \frac{1+\gamma}{2}j\pm\frac{1}{2}$, $k^+=\frac{1+\gamma}{2}k$, $l^+= \frac{1+\gamma}{2}l\mp \frac{1}{2}$, $j^{+}+k^{+}+l^{+}\in \N$, $k+l>j$ and $j+k>l$, then $(j^{+},k^{+},l^{+})$ and $(j-j^+,k-k^+,l-l^+)$ are admissible.
\label{dwa}
\end{lm}
\begin{proof}
 As previously, we denote $j^{-}:=j-j^+,k^{-}:=k-k^+,l^{-}:=l-l^+$ (it is easy to check, that they are nonnegative).

Let us check triangle inequalities:
\[
 j^++k^+=\frac{1+\gamma}{2}j\pm\frac{1}{2}+ \frac{1+\gamma}{2}k= \frac{1+\gamma}{2}(j+k)\pm\frac{1}{2}> \frac{1+\gamma}{2}l\pm\frac{1}{2}=l^+\mp 1
\]

By arguments used in previous lemma, we obtain $j^++k^+-l^+\geq 0$.

Let us check another inequality:
\[
 k^++l^+=\frac{1+\gamma}{2}k+\frac{1+\gamma}{2}l\mp\frac{1}{2}= \frac{1+\gamma}{2}(k+l)\mp\frac{1}{2}> \frac{1+\gamma}{2}j\mp\frac{1}{2}=j^+\pm 1.
\]
As a result $k^++l^+-j^+\geq 0$.

Finally
\[
 j^++l^+=\frac{1+\gamma}{2}j\pm\frac{1}{2}+\frac{1+\gamma}{2}l\mp\frac{1}{2}= \frac{1+\gamma}{2}(j+l) \geq \frac{1+\gamma}{2}k=k^+
\]
This finishes the prove of triangle inequalities. Proof for $j^-,k^-,l^-$ is the same.
\end{proof}
\begin{lm}
   Let $(j,k,l)$ be admissible and $j\not=0,l\not= 0$. If $j^{+},k^{+},l^{+}$ are elements of $\NN$ satisfying:
$j^+ = \frac{1+\gamma}{2}j\pm\frac{1}{2}$, $k^+=\frac{1+\gamma}{2}k$, $l^+= \frac{1+\gamma}{2}l\pm \frac{1}{2}$, $j^{+}+k^{+}+l^{+}\in \N$, $k+l=j$ or $j+k=l$, then $(j^{+},k^{+},l^{+})$ and $(j-j^+,k-k^+,l-l^+)$ are admissible.
\label{trzy}
\end{lm}
\begin{proof}
 Let $k+l=j$. Then $k^++l^+=j^+$ which proves triangle inequalities. The proof is the same for $j+k=l$. One checks in the same way that $(j-j^+,k-k^+,l-l^+)$ is admissible.
\end{proof}
\subsection{The fact that $\chi\not=0$}\label{chi_not_0} 
The fact that $\chi\not=0$ was proved in our previous paper \cite{cEPRL}. Here we recall only the result.
\begin{lm}\label{oldlemma}
 Let $(j^+,k^+,l^+)$,$(j^-,k^-,l^-)$ be admissible. Define $j=j^++j^-$, $k=k^++k^-$, $l=l^++l^-$. Take any non-zero $\eta\in{\rm Inv}\left(\Hil_{j}\otimes\Hil_{k}\otimes\Hil_{l}^*\right)$, $\eta^+\in{\rm Inv}\left(\Hil_{j^+}^*\otimes\Hil_{k^+}^*\otimes\Hil_{l^+}\right)$ and $\eta^-\in{\rm Inv}\left(\Hil_{j^-}^*\otimes\Hil_{k^-}^*\otimes\Hil_{l^-}\right)$. We have:
\[
 {\eta^+}_{A^+ B^+}^{C^+} {\eta^-}_{A^- B^-}^{C^-}  C^{A^+ A^-}_{A} C^{B^+ B^-}_{B} \eta^{A B}_C = \chi\  C^{C^+ C^-}_{C}
\]
for $\chi\not=0$.
\label{nonzero}
\end{lm}
Interestingly, this lemma may be proved also using argument different from the one used in \cite{cEPRL}. Now we present it.

First notice, it is enough to show, that, under assumptions above,
\[
 {\eta^+}_{A^+ B^+}^{C^+} {\eta^-}_{A^- B^-}^{C^-}  C^{A^+ A^-}_{A} C^{B^+ B^-}_{B} \eta^{A B}_C  C_{C^+ C^-}^{C}\not=0
\]
for some non-zero $C_{l^+ l^-}^{l}$.

However the expression  ${\eta^+}_{A^+ B^+}^{C^+} {\eta^-}_{A^- B^-}^{C^-}  C^{A^+ A^-}_{A} C^{B^+ B^-}_{B} \eta^{A B}_C  C_{C^+ C^-}^{C}$ is proportional with non-zero proportionality factor to 9j-symbol, i.e.:
\begin{displaymath}
{\eta^+}_{A^+ B^+}^{C^+} {\eta^-}_{A^- B^-}^{C^-}  C^{A^+ A^-}_{A} C^{B^+ B^-}_{B} \eta^{A B}_C  C_{C^+ C^-}^{C}=
\lambda \left\lbrace   \begin{array}{ccc}
j^- & l^- & k^- \\
j^+ & l^+ & k^+ \\
j & l & k
\end{array} \right\rbrace,
\end{displaymath}
where $\lambda\not=0$. The appearance of this $9j$-symbol here is strictly connected with the expansion of fusion coefficient into product of $9j$-symbols done in four-valent case in the article \cite{asymptotics_fusion}. From the properties of $9j$-symbol and admissibility of $(j^+,k^+,l^+)$, $(j^-,k^-,l^-)$ follows that this $9j$-symbol is proportional to a $3j$-symbol (see e.g. equation (37) in \cite{asymptotics_fusion}) with non-zero proportionality constant, i.e.:
\begin{displaymath}
\left\lbrace   \begin{array}{ccc}
j^- & l^- & k^- \\
j^+ & l^+ & k^+ \\
j & l & k
\end{array} \right\rbrace= \mu \left(   \begin{array}{ccc}
l^- & l^+ & l \\
j^--k^- & j^+-k^+ & -(j-k) \\
\end{array} \right),
\end{displaymath}
where $\mu\not=0$.

Recall that $l=l^++l^-$, so 
\begin{align*}
 &\left(   \begin{array}{ccc}
l^- & l^+ & l \\
j^--k^- & j^+-k^+ & -(j-k) \\
\end{array} \right)=\\= (-1)^{l^--l^++j-k}\left[\frac{(2 l^-)!(2l^+)!)}{(2l+1)!}\right. & \left.\frac{(l+j-k)!(l-j+k)!}{(l^-+j^--k^-)!(l^--j^-+k^-)!(l^++j^+-k^+)!(l^+-j^++k^+)!}\right]^{\frac{1}{2}}.
\end{align*}
From admissibility of $(j^+,k^+,l^+)$, $(j^-,k^-,l^-)$ follows that $\left(   \begin{array}{ccc}
l^- & l^+ & l \\
j^--k^- & j^+-k^+ & -(j-k) \\
\end{array} \right)\not=0$. Finally:
\[
 {\eta^+}_{A^+ B^+}^{C^+} {\eta^-}_{A^- B^-}^{C^-}  C^{A^+ A^-}_{A} C^{B^+ B^-}_{B} \eta^{A B}_C  C_{C^+ C^-}^{C}=\lambda\mu\left(   \begin{array}{ccc}
l^- & l^+ & l \\
j^--k^- & j^+-k^+ & -(j-k) \\
\end{array} \right)\not=0.
\]

\section{Proof of the theorem}\label{inductive_hyp}


\subsection{The inductive hypothesis}
We base our prove on the following inductive hypothesis for $n\geq 3$, $n\in \N$:
\begin{quote}
{\bf Hyp $n$:} Suppose that $(k_1,\ldots, k_n)$, $(j^\pm_1,\ldots, j^\pm_n)$ satisfy condition {\bf Con $n$} and that $\omicron\in 
{\rm Inv}\left(\Hil_{k_1}\otimes\cdots\otimes\Hil_{k_n}\right)$. Then, there 
exists 
$$\phi\in {\rm Inv}\left(\Hil_{j^+_1}\otimes\cdots\otimes\Hil_{j^+_n}\right)\otimes {\rm Inv}\left(\Hil_{j^-_1}\otimes\cdots\otimes\Hil_{j^-_n}\right)$$ such that
$\langle \phi, \iota_{k_1\ldots k_n}(\omicron)\rangle\not=0$.
\end{quote}
Notice that we do not restrict to tree-irreducible intertwiners anymore. As mentioned before, this proves injectivity of the EPRL map (theorem \ref{inj-EPRL}) for $n\geq 3$. One needs to check cases $n=1$ and $n=2$ separately but this is straightforward.


\subsection{Proof}\label{full_proof}
Previously we restricted ourselves to tree-irreducible intertwiners, because then the lowest spin $k_{\alpha}$ in the decomposition \eqref{inv_dec} (we denote it by $k'_{\alpha}$) as well as $k_1$ are different than 0, if $n>3$. In general $k'_{\alpha}$ or $k_1$ may be equal 0 for $n>3$ and then our procedure determining $j_{\alpha}^+$ and $j_{\alpha}^-$ may not be applied (lemmas \ref{jeden}, \ref{dwa}, \ref{trzy} require $k_{\alpha}\not=0$, $k_1\not=0$). Actually the case $k'_{\alpha}=0$ (so $k_1=k_2$) and $j_1^+=j^+_2$ is not problematic -- we simply take $j_{\alpha}^+ = 0$ and follow steps 3-5 in section \ref{simpl_ind_step}. The case $k_1=0$ is also simple, because $j_1^++j_1^-=k_1=0$ implies $j_1^\pm=0$ and the inductive step is trivial. Problems appear, when $k'_{\alpha}=0$ and $j_{1}^{+}=j^{+}_2\pm\frac{1}{2}$, $j_{1}^{-}=j^{-}_2\mp\frac{1}{2}$. We treat this case separately.

The inductive step we start (as in simplified case in section \ref{simpl_proof}) by expanding $\omicron$ as in equation \eqref{inv_dec} and finding minimal $k_{\alpha}$ which we call $k'_{\alpha}$. We may perform standard procedure unless we are in the problematic case. Note that in this case 
\begin{equation}
 j^+_i\in \N+\frac{1}{2},\ \ i>1,
\end{equation}
as $j^+_1=j^+_2\pm \frac{1}{2}$ and sequences are ordered.
If this is the case, we expand the intertwiner $\omicron$ one level further, i.e. instead of formula \eqref{inv_dec} we use the following one:
\begin{equation}
 \label{inv_dec_level_2}
 \omicron^{A_1 A_2 \ldots A_n}=\sum_{(k_{\alpha},k_{\beta}) \in K} C^{A_1 A_2}_{A_\alpha} C^{A_\alpha A_3}_{A_\beta} (\omicron^{k_{\alpha} k_{\beta}})^{A_{\beta} A_4 \ldots A_n},
\end{equation}
 where $K:=\{(k_{\alpha},k_{\beta})\in \NN\times\NN: \omicron^{k_{\alpha} k_{\beta}}\not\equiv 0 \}$. We define $K'=K\cap \{(k_{\alpha},k_{\beta}): k_{\beta}<k_3 \}$. There are two cases $K'=\emptyset$ and $K'\not=\emptyset$ which we describe in next two sections. The procedure is summarised by the diagram below.

Importantly notice that in case $n=4$, we either obtain $k_{\alpha}'>0$ or $k_{\alpha'}=0$, $j_1^+=j_2^+$. In this case notice that either $j_1^+\in \N$, $j_2^+\in \N$ or $j_1^+\in \N+\frac{1}{2}$, $j_2^+\in \N+\frac{1}{2}$ (this follows from the fact that $j_1^++j_2^++j_3^++j_4^+\in \N$ and from ordering of $k_i$) -- as a result if $k_{\alpha}'=0$ then $j_1^+=j_2^+$. This means that when $n=4$, the inductive step from simplified case may be used. As a result the check of initial conditions done in the proof of tree-irreducible case is sufficient in general case presented here.

\begin{center}
\vspace{0.5cm}
\small
\psframebox[linearc=0.5,cornersize=absolute,framesep=10pt]{%
  \psset{shadowcolor=black!70,blur=true}%
  \begin{psmatrix}[rowsep=0.4,colsep=0.5]
\psframebox[shadow=true]{\begin{tabular}{c} Ordering of $k_1,\ldots,k_n$ such that first are the $k_i$ \\ with $j_i^+\in \N$ and then $k_i$ with $j_i^+\in \N+\frac{1}{2}$. \end{tabular}
}\\
   \psframebox[fillstyle=solid,fillcolor=Pink,shadow=true]{$k_1=0$ ?} &
       \psframebox[shadow=true]{\begin{tabular}{c} It follows, that $j_1^\pm=0$ and\\ inductive step is trivial \end{tabular}
} \\
    \psframebox[fillstyle=solid,fillcolor=Pink,shadow=true]{$k'_{\alpha}>0$ ?} &
       \psframebox[shadow=true]{\begin{tabular}{c} Follow inductive step in \\ simplified case -- section \ref{simpl_ind_step} \end{tabular}
} \\
    \psframebox[fillstyle=solid,fillcolor=Pink,shadow=true]{$j^+_1=j_2^+$} &
       \psframebox[shadow=true]{\begin{tabular}{c} Take $j_{\alpha}^{\pm}=0$ and follow point 3-5 \\ in inductive step of simplified case \end{tabular}
} \\
    \psframebox[fillstyle=solid,fillcolor=Pink,shadow=true]{$k_2\not= k_3$} &
       \psframebox[shadow=true]{\begin{tabular}{c} Exchange $k_2$ and $k_3$ and follow \\ inductive step in simplified case \end{tabular}
}
    ~\\
    \psframebox[shadow=true,fillstyle=solid,fillcolor=Pink]{$K'\not=\emptyset$ ? (we have $j^+_1=j_2^+\pm\frac{1}{2}$)} & \psframebox[shadow=true]{Follow steps in section \ref{K_nonempty}} \\
    \psframebox[shadow=true]{Follow steps in section \ref{K_empty}}
    \ncline{->}{2,1}{3,1}<{\textcolor{red}{No}}
    \ncline{->}{3,1}{4,1}<{\textcolor{red}{No}}
    \ncline{->}{3,1}{3,2}^{\textcolor{red}{\hspace{-60px} Yes}}
    \ncline{->}{4,1}{4,2}^{\textcolor{red}{\hspace{-60px} Yes}}
    \ncline{->}{6,1}{6,2}^{\textcolor{red}{Yes}}
    \ncline{->}{5,1}{5,2}^{\textcolor{red}{\hspace{-60px} Yes}}
    \ncline{->}{4,1}{5,1}<{\textcolor{red}{No}}
    \ncline{->}{5,1}{6,1}<{\textcolor{red}{No}}
    \ncline{->}{6,1}{7,1}<{\textcolor{red}{No}}
    \ncline{->}{2,1}{2,2}^{\textcolor{red}{\hspace{-60px} Yes}}
    \ncline{->}{1,1}{2,1}
    \ncbar[angleA=-90,armB=0,nodesepB=0.25]{->}{3,3}{4,1}
    \end{psmatrix}%
}
\end{center}


\subsection{The case $K'\not=\emptyset$}\label{K_nonempty}
\begin{enumerate}
\item Find $k''_{\alpha}$ and $k_{\beta}'$, such that:
\begin{equation}
 k_\alpha''=\min\{k_\alpha\colon \exists k_\beta,\ (k_\alpha,k_\beta)\in K'\}
\end{equation}
and
\begin{equation}\label{minimal_k_beta}
 k_\beta'=\min\{k_\beta\colon (k_\alpha',k_\beta)\in K'\}
\end{equation}
They exist, because $K'$ is non-empty.
 \item Notice that $k_{\alpha}''>0$ because $(k_3,k_\alpha'',k_\beta')$ is admissible and $k_\beta'<k_3$. We define $j_{\alpha}^{\pm}$ using the procedure from section \ref{simpl_ind_step}. 

If $k_{\beta}'>0$, we use the same procedure (but for triple $(k_\alpha'',k_3,k_\beta')$) to define $j_{\beta}^{\pm}$ and if $k_{\beta}=0$, we take $j_{\beta}^{\pm}=0$. Let us check know that $j_{\alpha}^{\pm}$ and $j_{\beta}^{\pm}$ is a good choice.
\begin{itemize}
 \item $(j_1^{\pm},j_2^{\pm},j_{\alpha}^{\pm})$ are admissible -- this is guaranteed by procedure from section \ref{simpl_ind_step};
 \item $(j_{\alpha}^{\pm},j_3^{\pm},j_\beta^{\pm})$ are admissible: 

If $k_{\beta}'>0$ then this is guaranteed by procedure from section \ref{simpl_ind_step}. 

If $k_{\beta}'=0$, then $k_{\alpha}''=k_3 $. As a result we have $j_3^{+} -\frac{1}{2}\leq j_\alpha^{+}\leq j_3^{+}+\frac{1}{2}$. However $j_{\alpha}^+\in \N+\frac{1}{2}$ ($j_{1}^{+}=j_{2}^{+}+\frac{1}{2}$ or $j_{1}^{+}=j_{2}^{+}-\frac{1}{2}$, so $j_1^{+}+j_2^{+}$ is not an integer). From the ordering $j^+_3\in \N+\frac{1}{2}$. Finally we have $j_{\alpha}^{+}=j_3^{+}$ and $j_{\alpha}^{-}=k''_{\alpha}-j_{\alpha}^{+}=k_3-j_3^{+}=j_3^{-}$. Obviously $(j_{\alpha}^{\pm},j_3^{\pm},0)$ are admissible.
\item $j^{\pm}_{\beta}+j_4^{\pm}+\ldots+j_n^{\pm}\in \N$

We know that $j_1^{+}\in \N$, $j_2^{+}\in \N+\frac{1}{2}$,$j_3^{+}\in \N+\frac{1}{2}$, so $j_{\beta}^{+}\in \N$ and $j_{1}^{+}+j_{2}^{+}+j_{3}^{+}\in \N$. Finally from $j^{+}_{1}+j_2^{+}+ j_3^{+}+j_{4}^{+}\ldots+j_n^{+}\in \N$, follows that $j_{\beta}^{+}+j_{4}^{+}\ldots+j_n^{+}\in \N$.

Using the facts that $j_{\beta}^{-}+j_{4}^{-}\ldots+j_n^{-}=k'_{\beta}+k_{4}\ldots+k_n - (j_{\beta}^{+}+j_{4}^{+}\ldots+j_n^{+})$ and $k'_{\beta}+k_4+\ldots+k_n\in \N$, we obtain $j_{\beta}^{-}+j_{4}^{-}\ldots+j_n^{-}\in \N$.
\item  We see that $\frac{1+\gamma}{2} k'_{\beta} - \frac{1}{2} \leq j_{\beta}^+\leq \frac{1+\gamma}{2} k'_{\beta} + \frac{1}{2}$. We have also $j^+_4\in \N+\frac{1}{2}$ and so the ordering property is satisfied.
\end{itemize}
Eventually, {\bf Con $n-2$} is fulfilled for $(k'_{\beta},k_4,\ldots,k_n)$ and $(j_{\beta}^{\pm},j_4^{\pm},\ldots,j_n^{\pm})$. 
\item 
From {\bf Hyp $n-2$} follows that for $\omicron^{k''_{\alpha} k'_{\beta}}$ from \eqref{inv_dec_level_2} there exists \[\phi^{k''_{\alpha} k'_{\beta}}\in {\rm Inv}\left(\Hil_{j^+_{\alpha}}\otimes\Hil_{j^+_{\beta}}\otimes\cdots\otimes\Hil_{j^+_n}\right)\otimes {\rm Inv}\left(\Hil_{j^-_{\alpha}}\otimes\Hil_{j^+_{\beta}}\otimes\cdots\otimes\Hil_{j^-_n}\right),\] such that 
\[\langle \phi^{k''_{\alpha} k'_{\beta}}, \iota'_{k'_\beta \ldots k_n}(\omicron^{k''_{\alpha} k'_{\beta}})\rangle\not=0.\]
\item Having defined $\phi^{k''_{\alpha} k'_{\beta}}$, we construct $\phi$:
\[
 \phi^{A_1^+\ldots A_n^+, A_1^{-}\ldots A_n^{-}}:=C^{A_1^+ A_2^+}_{A_{\alpha}^+} C^{A_{\alpha}^+ A_3^+}_{A_{\beta}^+} C^{A_1^- A_2^-}_{A_{\alpha}^-} C^{A_{\alpha}^- A_3^-}_{A_{\beta}^-} (\phi^{k''_{\alpha} k'_{\beta}})^{A_{\beta}^+ A_4^+\ldots A_n^+, A_{\beta}^- A_4^{-}\ldots A_n^{-}}
\]
\item The $\phi$ constructed in previous point is the $\phi$ we are looking for, i.e.
$\langle \phi, \iota'_{k_1\ldots k_n}(\omicron)\rangle\not=0$. We now prove this statement.
\begin{enumerate}
 \item First, using equation \eqref{inv_dec_level_2} write $\langle \phi, \iota'_{k_1\ldots k_n}(\omicron)\rangle$ as a sum:
\begin{equation} \label{level_2_sum}
 \langle \phi, \iota'_{k_1\ldots k_n}(\omicron)\rangle = \sum_{(k_{\alpha}, k_{\beta})\in K} \langle \phi, \iota'_{k_1\ldots k_n}(C^{k_1 k_2}_{k_\alpha}\circ C^{k_\alpha k_3}_{k_\beta}\circ \omicron^{k_{\alpha} k_{\beta}})\rangle,
\end{equation}
where ${(C^{k_1 k_2}_{k_\alpha}\circ C^{k_\alpha k_3}_{k_\beta} \circ \omicron^{k_{\alpha} k_{\beta}})}^{A_1 A_2 \ldots A_n}:= C^{A_1 A_2}_{A_\alpha} {C}^{A_\alpha A_3}_{A_\beta} (\omicron^{k_{\alpha} k_{\beta}})^{A_{\beta} A_4 \ldots A_n} $.


\item Let us compute $\langle \phi, \iota'_{k_1\ldots k_n}(C^{k_1 k_2}_{k_\alpha}\circ C^{k_\alpha k_3}_{k_\beta}\circ \omicron^{k_{\alpha} k_{\beta}})\rangle$ (see fig. \ref{K_nonempty}):
\begin{figure}
 \includegraphics[width=\textwidth]{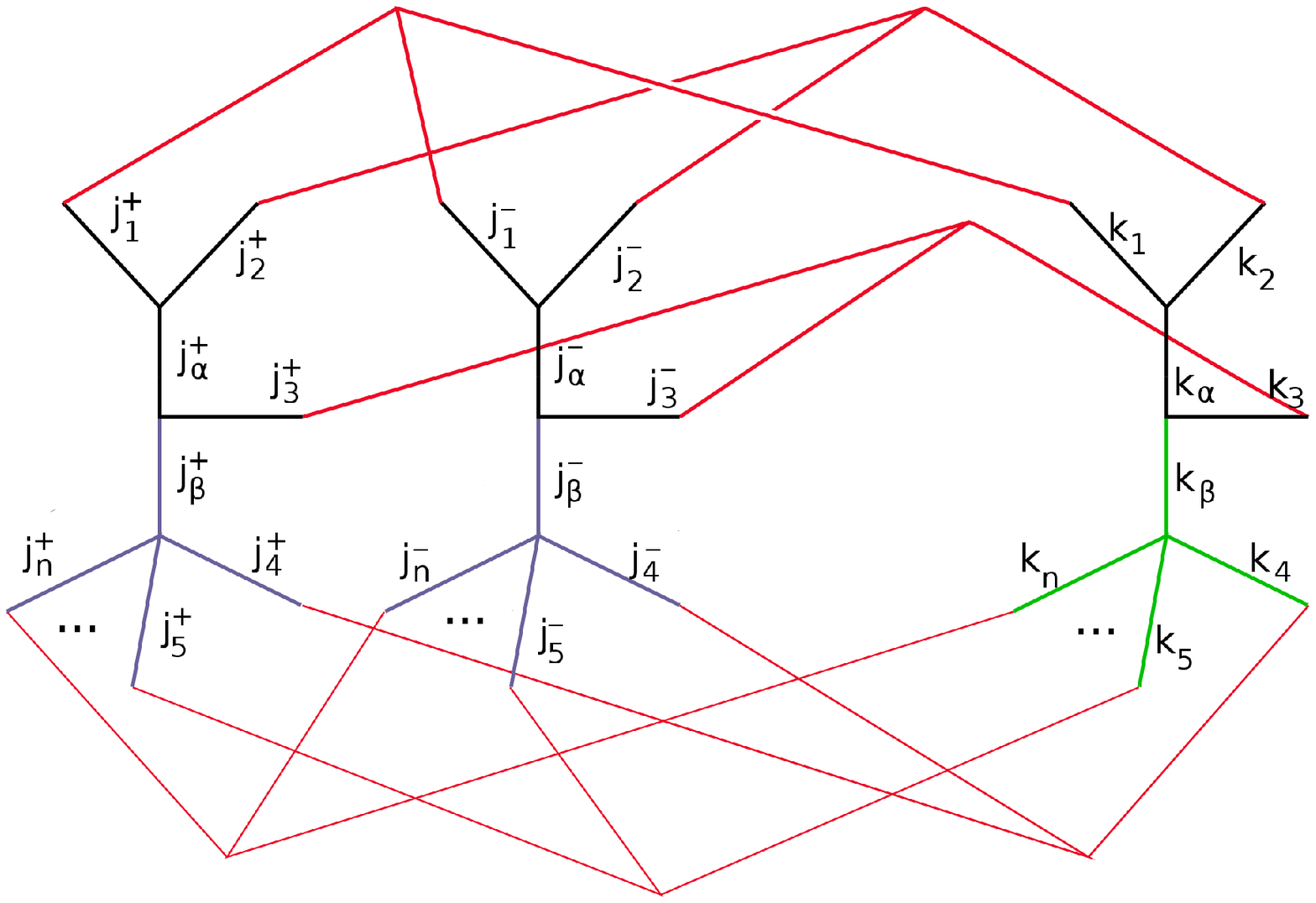}
\caption{Schematic picture of a term in the sum \eqref{level_2_sum}, i.e.
$
 \langle \phi, {\color{czerwony} \iota'_{k_1\ldots k_n}}(C^{k_1 k_2}_{k_\alpha}\circ C^{k_\alpha k_3}_{k_\beta}\circ {\color{zielony} \omicron^{k_{\alpha} k_{\beta}}})\rangle = \langle C^{j_1^+ j_2^+}_{j_{\alpha}^+}\circ C^{j_{\alpha}^+ j_3^+}_{j_{\beta}^+} \circ C^{j_1^- j_2^-}_{j_{\alpha}^-}\circ C^{j_{\alpha}^- j_3^-}_{j_{\beta}^-} \circ {\color{niebieski} \phi^{k''_{\alpha} k'_{\beta}}}, {\color{czerwony}\iota'_{k_1\ldots k_n}}(C^{k_1 k_2}_{k_\alpha}\circ C^{k_\alpha k_3}_{k_\beta}\circ {\color{zielony} \omicron^{k_{\alpha} k_{\beta}}})\rangle 
$}
\label{fig:K_nonempty}
\end{figure}

\begin{align*}
\langle \phi, \iota'_{k_1\ldots k_n}(C^{k_1 k_2}_{k_\alpha}\circ C^{k_\alpha k_3}_{k_\beta}\circ \omicron^{k_{\alpha} k_{\beta}})\rangle=
({\phi}^{k''_{\alpha} k'_{\beta}})^\dagger_{A_{\beta}^+ A_4^+\ldots A_n^+, A_{\beta}^- A_4^{-}\ldots A_n^{-}} C_{A_1^+ A_2^+}^{A_{\alpha}^+} & C_{A_{\alpha}^+ A_3^+}^{A_{\beta}^+} C_{A_1^- A_2^-}^{A_{\alpha}^-} C_{A_{\alpha}^- A_3^-}^{A_{\beta}^-}\\ C_{A_1}^{A_1^+ A_1^-} \ldots C_{A_n}^{A_n^+ A_n^-} C^{A_1 A_2}_{A_\alpha} {C}^{A_\alpha A_3}_{A_\beta} (\omicron^{k_{\alpha} k_{\beta}})^{A_{\beta} A_4 \ldots A_n} = {\color{blue} C_{A_1^+ A_2^+}^{ A_{\alpha}^+}  C_{A_1^- A_2^-}^{A_{\alpha}^-} C_{A_1}^{A_1^+ A_1^-} }&{\color{blue} C_{A_2}^{A_2^+ A_2^-} C^{A_1 A_2}_{A_\alpha}} \\ {\color{red} C_{A_{\alpha}^+ A_3^+}^{A_{\beta}^+} C_{A_{\alpha}^- A_3^-}^{A_{\beta}^-}  C_{A_3}^{A_3^+ A_3^-}  {C}^{{A_\alpha} A_3}_{A_\beta}} ({\phi}^{k''_{\alpha} k'_{\beta}})^\dagger_{A_{\beta}^+ A_4^+\ldots A_n^+, A_{\beta}^- A_4^{-}\ldots A_n^{-}} C_{A_4}^{A_4^+ A_4^-} \ldots C_{A_n}^{A_n^+ A_n^-} & (\omicron^{k_{\alpha} k_{\beta}})^{A_{\beta} A_4 \ldots A_n}
\end{align*}

Using lemma \ref{nonzero} we obtain, that for some $\xi_1\not=0$
\[ {\color{blue} C_{A_1^+ A_2^+}^{A_{\alpha}^+} C_{A_1^- A_2^-}^{A_{\alpha}^-} C_{A_1}^{A_1^+ A_1^-} C_{A_2}^{A_2^+ A_2^-} C^{A_1 A_2}_{A_\alpha} } =\left\{\begin{array}{ll}
                    0, & k_\alpha>j_\alpha^++j^-_\alpha\\
                     \xi_1 {\color{DarkOrchid} C_{A_\alpha}^{A_\alpha^+ A_\alpha^-}}, & k_\alpha=j_\alpha^++j^-_\alpha
                   \end{array}\right.
\] 
Applying this lemma second time we get, that for some $\xi_2\not=0$
\[ {\color{red} C_{A_{\alpha}^+ A_3^+}^{A_{\beta}^+} C_{A_{\alpha}^- A_3^-}^{A_{\beta}^-} }{\color{DarkOrchid} C_{A_\alpha}^{A_\alpha^+ A_\alpha^-}}{\color{red} C_{A_3}^{A_3^+ A_3^-}  {C}^{{A_\alpha} A_3}_{A_\beta}} =\left\{\begin{array}{ll}
                    0, & k_\beta>j_\beta^++j^-_\beta\\
                    \xi_2 C_{A_\beta}^{A_\beta^+ A_\beta^-}, & k_\beta=j_\beta^++j^-_\beta
                   \end{array}\right.\]
Finally, for $\xi:=\xi_1 \xi_2 \not= 0$ :
\begin{equation}\label{level_2_eq}
 \langle \phi, \iota'_{k_1\ldots k_n}(C^{k_1 k_2}_{k_\alpha}\circ C^{k_\alpha k_3}_{k_\beta} \omicron^{k_{\alpha} k_{\beta}})\rangle=\left\{\begin{array}{ll}
                    0, & k_\alpha>k_\alpha''\ {\rm or}\ k_\beta>k_\beta'\\
                    \xi\langle \phi^{k''_{\alpha} k'_{\beta}}, \iota'_{k'_\beta \ldots k_n}(\omicron^{k''_{\alpha} k'_{\beta}})\rangle, & k_\alpha=k_\alpha'\ {\rm and}\ k_\beta=k_\beta'\\
  *, & {\rm otherwise}
                   \end{array}\right.
\end{equation}
\item Now we use formula just obtained \eqref{level_2_eq}  to calculate the sum \eqref{level_2_sum}.

First notice that $k_\beta'\leq k_3$. As a result the elements in the sum \eqref{level_2_sum} with $k_\beta>k_3$ are vanishing and the sum is actually over the $K'$:
\[
 \langle \phi, \iota'_{k_1\ldots k_n}(\omicron)\rangle = \sum_{(k_{\alpha}, k_{\beta})\in K'} \langle \phi, \iota'_{k_1\ldots k_n}(C^{k_1 k_2}_{k_\alpha}\circ C^{k_\alpha k_3}_{k_\beta}\circ \omicron^{k_{\alpha} k_{\beta}})\rangle
\]
However from the definition of $k_\alpha''$ and $k_\beta'$ follows that
\[
 \langle \phi, \iota'_{k_1\ldots k_n}(\omicron)\rangle = \sum_{k_\alpha\geq k_\alpha'', k_\beta\geq k_\beta'} \langle \phi, \iota'_{k_1\ldots k_n}(C^{k_1 k_2}_{k_\alpha}\circ C^{k_\alpha k_3}_{k_\beta}\circ \omicron^{k_{\alpha} k_{\beta}})\rangle
\]
Finally, using \eqref{level_2_eq} we obtain:
\[
 \langle \phi, \iota'_{k_1\ldots k_n}(\omicron)\rangle = \xi \langle \phi^{k''_{\alpha} k'_{\beta}}, \iota'_{k'_\beta \ldots k_n}(\omicron^{k''_{\alpha} k'_{\beta}})\rangle
\]
and
\[
 \langle \phi, \iota'_{k_1\ldots k_n}(\omicron)\rangle \not= 0.
\]
\end{enumerate}
\end{enumerate}
\subsection{The case $K'=\emptyset$}\label{K_empty}
Let us change the basis used previously in the decomposition of $\omicron$ \eqref{inv_dec_level_2}:
\begin{equation}\label{inv_dec_level_2_empty}
 \omicron^{A_1 A_2 \ldots A_n}=\sum_{(k_{\widetilde\alpha},k_{\beta}) \in L} {C}^{A_2 A_3}_{A_{\widetilde \alpha}} C^{A_1 A_{\widetilde \alpha}}_{A_\beta}  (\omicron^{k_{\widetilde\alpha} k_{\beta}})^{A_{\beta} A_4 \ldots A_n},
\end{equation}
where $L:=\{(k_{\widetilde\alpha},k_{\beta})\in \NN\times\NN: \omicron^{k_{\widetilde\alpha} k_{\beta}}\not\equiv 0 \}$.

 We define:
\[
 L'=L\cap \{(k_{\widetilde \alpha},k_\beta): k_\beta=k_3 \}
\]
This set is non-empty, because $k_\beta=k_3$ was present in the decomposition \eqref{inv_dec_level_2}. In fact $(k_\alpha=0, k_\beta=k_3)\in K$ and so $k_\beta=k_3$ occurs also in the decomposition \eqref{inv_dec_level_2_empty}.
\begin{enumerate}
 \item Find $k'_{\widetilde \alpha}$ such that:
\[
 k'_{\widetilde \alpha}={\rm min}\{k_{\widetilde \alpha}: \exists k_\beta,\ (k_{\widetilde \alpha},k_\beta)\in L'\}
\]
Note that $k_{\beta}<k_3$ does not appear in this decomposition because they are absent in the decomposition \eqref{inv_dec_level_2}.
Note also that if we defined $k'_\beta$ in analogous way to \eqref{minimal_k_beta}, i.e. $k_\beta'=\min\{k_\beta\colon (k_{\widetilde \alpha}',k_\beta)\in L'\}$, we would obtain trivially $k_{\beta}'=k_3$. In this section one may think that $k_{\beta}'=k_3$. However we will not write this $k_{\beta}'$ explicitly.

 
\item

We now define $j_{\widetilde \alpha}^\pm$, $j_\beta^\pm$. Note, that only $j_{\beta}^\pm$ (but not $j^+_{\widetilde\alpha}$) has to be of special form to match {\bf Con $n-2$}. The requirements for $j_{\widetilde \alpha}^\pm$ may be limited to assure admissibility conditions and the condition that $j_{\widetilde \alpha}^++j_{\widetilde \alpha}^-=k_{\widetilde \alpha}$. We will use this freedom to define $j_{\widetilde \alpha}^\pm$, $j_\beta^\pm$.

Note that, in our case $(j_2^\pm,j_3^\pm,j_{{\widetilde\alpha}}^\pm)$ are admissible iff $j_{{\widetilde \alpha}}^{\pm}\leq 2 j_2^{\pm}$ and $j_{\widetilde \alpha}^\pm\in \N$ (because $j^+_2=j^+_3$). We also have $j_1^+=j_2^+\pm\frac{1}{2}$. The choices of $j^\pm_{\widetilde\alpha}$ and $j^\pm_\beta$ in this case are given be the following diagram.
\begin{center}
\vspace{0.5cm}
\small
\psframebox[linearc=0.5,cornersize=absolute,framesep=10pt]{%
  \psset{shadowcolor=black!70,blur=true}%
  \begin{psmatrix}[rowsep=0.4,colsep=0.5]
          & \psframebox[fillstyle=solid,fillcolor=Pink,shadow=true]{$j^+_1=j^+_2+\frac{1}{2}$?} &
\psframebox[fillstyle=solid,fillcolor=Pink,shadow=true]{$j^+_1=j^+_2-\frac{1}{2}$?}\\
\psframebox[fillstyle=solid,fillcolor=Pink,shadow=true]{$k_\alpha'<2k_2$?}&
   \psframebox[shadow=true]{\begin{tabular}{c}
      We define $j^\pm_\beta=j^\pm_2\pm\frac{1}{2}$
      and\\ take any  $j^\pm_{\widetilde\alpha}$ satisfying\\ 
      $j_{\widetilde \alpha}^+\leq 2 j_2^+$, $j_{\widetilde \alpha}^-\leq 2 j_2^- - 1$,\\ $j_{\widetilde \alpha}^\pm\in\N$, $j_{\widetilde \alpha}^++j_{\widetilde \alpha}^-=k'_{\widetilde \alpha}$
      \end{tabular}} &
   \psframebox[shadow=true]{\begin{tabular}{c}
      We define $j^\pm_\beta=j^\pm_2\mp\frac{1}{2}$
      and\\ take any  $j^\pm_{\widetilde\alpha}$ satisfying\\ 
      $j_{\widetilde \alpha}^+\leq 2 j_2^+-1$, $j_{\widetilde \alpha}^-\leq 2 j_2^-$,\\ $j_{\widetilde \alpha}^\pm\in\N$, $j_{\widetilde \alpha}^++j_{\widetilde \alpha}^-=k'_{\widetilde \alpha}$
      \end{tabular}} \\
\psframebox[fillstyle=solid,fillcolor=Pink,shadow=true]{$k_\alpha'<2k_2$?}&
   \psframebox[shadow=true]{\begin{tabular}{c}
    We define $j_\beta^\pm=j^\pm_2\mp \frac{1}{2}$\\
    and $j^\pm_{\widetilde\alpha}=2j^\pm_2$	
                                                            \end{tabular}
}&
   \psframebox[shadow=true]{\begin{tabular}{c}
    We define $j_\beta^\pm=j^\pm_2\pm \frac{1}{2}$\\
    and $j^\pm_{\widetilde\alpha}=2j^\pm_2$	
                                                            \end{tabular}
}
 \end{psmatrix}%
}
\end{center}
Let us justify this choice. Suppose that $j_1^+=j_2^++\frac{1}{2}$ (the case $j_1^+=j_2^+-\frac{1}{2}$ is analogous).

\begin{itemize}
 \item {\bf Case of $k'_{\widetilde \alpha} < 2 k_2$.}
  Note that $k_2\not=0$. As a result $j_2^-\geq \frac{1}{2}$ ($|\gamma|<1$) and there exist $j_{\widetilde \alpha}^\pm$, such that $j_{\widetilde \alpha}^+\leq 2 j_2^+$, $j_{\widetilde \alpha}^-\leq 2 j_2^- - 1$, $j_{\widetilde \alpha}^\pm\in \N$. It is possible to choose $j_{\widetilde\alpha}^\pm$ satisfying $j_{\widetilde\alpha}^++j_{\widetilde\alpha}^-=k'_{\widetilde \alpha}$, because $j_{\widetilde \alpha}^++j_{\widetilde \alpha}^-\leq 2 (j_2^++j_2^-) -1 \Rightarrow j_{\widetilde \alpha}^++j_{\widetilde \alpha}^-< 2k_2$ (and $k_{\widetilde\alpha}'<2k_2$).

It is straightforward to check that $(j_2^\pm,j_3^\pm,j_{{\widetilde \alpha}}^\pm)$, $(j_1^\pm,j_{{\widetilde \alpha}}^\pm,j_{{\beta}}^\pm)$ are admissible:
\begin{equation}
0=|j_2^\pm-j_3^\pm|\leq j^\pm_{\widetilde\alpha}\leq j^\pm_2+j^\pm_3=2j^\pm_2,\quad
0=|j_1^\pm-j_\beta^\pm|\leq j^\pm_{\widetilde\alpha}\leq j^\pm_1+j^\pm_\beta
\end{equation}
but $j^+_1+j^+_\beta=2j^+_2$ and $j^-_1+j^-_\beta=2j^-_2-1$.

\item {\bf Case of $k'_{\widetilde \alpha}= 2 k_2$.}

As previously pointed out if $k_2\not=0$, then $j_2^\pm\geq \frac{1}{2}$. It follows that $2j^\pm_2\geq 1$. So $j^\pm_{\widetilde\alpha}\geq 1$ and $j^\pm_{\widetilde\alpha}\in \N$, $j^+_{\widetilde\alpha}+j^-_{\widetilde\alpha}=k_{\widetilde\alpha}'$.

It is straightforward to check that $(j_2^\pm ,j_3^\pm,j_{{\widetilde \alpha}}^\pm)$, $(j_1^\pm,j_{{\widetilde \alpha}}^\pm,j_{{\beta}}^\pm)$ are admissible:
\begin{equation}
 0=|j_2^\pm-j_3^\pm|\leq j^\pm_{\widetilde\alpha}\leq j^\pm_2+j^\pm_3=2j^\pm_2,\ \ 
1=|j_1^\pm-j_\beta^\pm|\leq j^\pm_{\widetilde\alpha}\leq j^\pm_1+j^\pm_\beta=2j^\pm_2
\end{equation}
\end{itemize}

\item 
 Note that because $j^+_1$ and $j^+_{\widetilde\alpha}$ are natural then also $j_{\beta}^+\in\N$. Recall also that $j_1^+\in\N$,$j_2^+, j_3^+\in\N+\frac{1}{2}$ and $j_1^++\ldots j_n^+\in \N$. As a result $j_\beta^++j_4^++\ldots j_n^+\in \N$ and $j_\beta^-+j_4^-+\ldots j_n^-=k_3+k_4+\ldots +k_n - (j_\beta^++j_4^++\ldots j_n^+)\in \N$.

We have also that $\frac{1+\gamma}{2} k_3 - \frac{1}{2} \leq j_{\beta}^+\leq \frac{1+\gamma}{2} k_3 + \frac{1}{2}$ and $j^+_4\in \N+\frac{1}{2}$, so {\bf Con $n-2$} is fulfilled for $(k_3,k_4,\ldots,k_n)$ and $(j_{\beta}^{\pm},j_4^{\pm},\ldots,j_n^{\pm})$. 


\item

From {\bf Hyp $n-2$} follows that for $\omicron^{k'_{\widetilde \alpha} k_3}$ from \eqref{inv_dec_level_2_empty} there exists \[\phi^{k'_{\widetilde\alpha} k_3}\in {\rm Inv}\left(\Hil_{j^+_{\alpha}}\otimes\Hil_{j^+_{\beta}}\otimes\cdots\otimes\Hil_{j^+_n}\right)\otimes {\rm Inv}\left(\Hil_{j^-_{\alpha}}\otimes\Hil_{j^+_{\beta}}\otimes\cdots\otimes\Hil_{j^-_n}\right),\] such that 
\[\langle \phi^{k'_{\widetilde\alpha} k_3}, \iota'_{k_3 \ldots k_n}(\omicron^{k'_{\widetilde \alpha} k_3})\rangle\not=0.\]


\item Having defined $\phi^{k'_{\widetilde\alpha} k_3}$, we construct $\phi$:
\[
 \phi^{A_1^+\ldots A_n^+, A_1^{-}\ldots A_n^{-}}:= C^{A_2^+ A_3^+}_{A_{\widetilde\alpha}^+} C^{A_{\widetilde\alpha}^+ A_1^+ }_{A_{\beta}^+} C^{A_2^- A_3^-}_{A_{\widetilde \alpha}^-} C^{A_{\widetilde\alpha}^- A_1^-}_{A_{\beta}^-} (\phi^{k'_{\widetilde\alpha} k_3})^{A_{\beta}^+ A_4^+\ldots A_n^+, A_{\beta}^- A_4^{-}\ldots A_n^{-}}
\]


\item The $\phi$ constructed in previous point is the $\phi$ we are looking for, i.e.
$\langle \phi, \iota'_{k_1\ldots k_n}(\omicron)\rangle\not=0$. We now prove this statement.
\begin{enumerate}
 \item First, using equation \eqref{inv_dec_level_2_empty} write $\langle \phi, \iota'_{k_1\ldots k_n}(\omicron)\rangle$ as a sum:
\begin{equation} \label{level_2_empty_sum}
 \langle \phi, \iota'_{k_1\ldots k_n}(\omicron)\rangle = \sum_{(k_{\widetilde\alpha}, k_{\beta})\in L} \langle \phi, \iota'_{k_1\ldots k_n}({C}^{k_2 k_3}_{k_{\widetilde \alpha}} \circ C^{k_1 k_{\widetilde \alpha}}_{k_\beta} \circ \omicron^{k_{\widetilde\alpha} k_{\beta}})\rangle,
\end{equation}
where ${({C}^{k_2 k_3}_{k_{\widetilde \alpha}} \circ C^{k_1 k_{\widetilde \alpha}}_{k_\beta} \circ \omicron^{k_{\widetilde\alpha} k_{\beta}})}^{A_1 A_2 \ldots A_n}:=  C^{A_2 A_3}_{A_{\widetilde \alpha}} {C}^{A_1 A_{\widetilde\alpha}}_{A_\beta} (\omicron^{k_{\widetilde\alpha} k_{\beta}})^{A_{\beta} A_4 \ldots A_n} $.

\item Let us compute $\langle \phi, \iota'_{k_1\ldots k_n}({C}^{k_2 k_3}_{k_{\widetilde \alpha}} \circ C^{k_1 k_{\widetilde \alpha}}_{k_\beta}\circ \omicron^{k_{\widetilde\alpha} k_{\beta}})\rangle$ (see fig. \ref{K_empty}):

\begin{figure}
 \includegraphics[width=\textwidth]{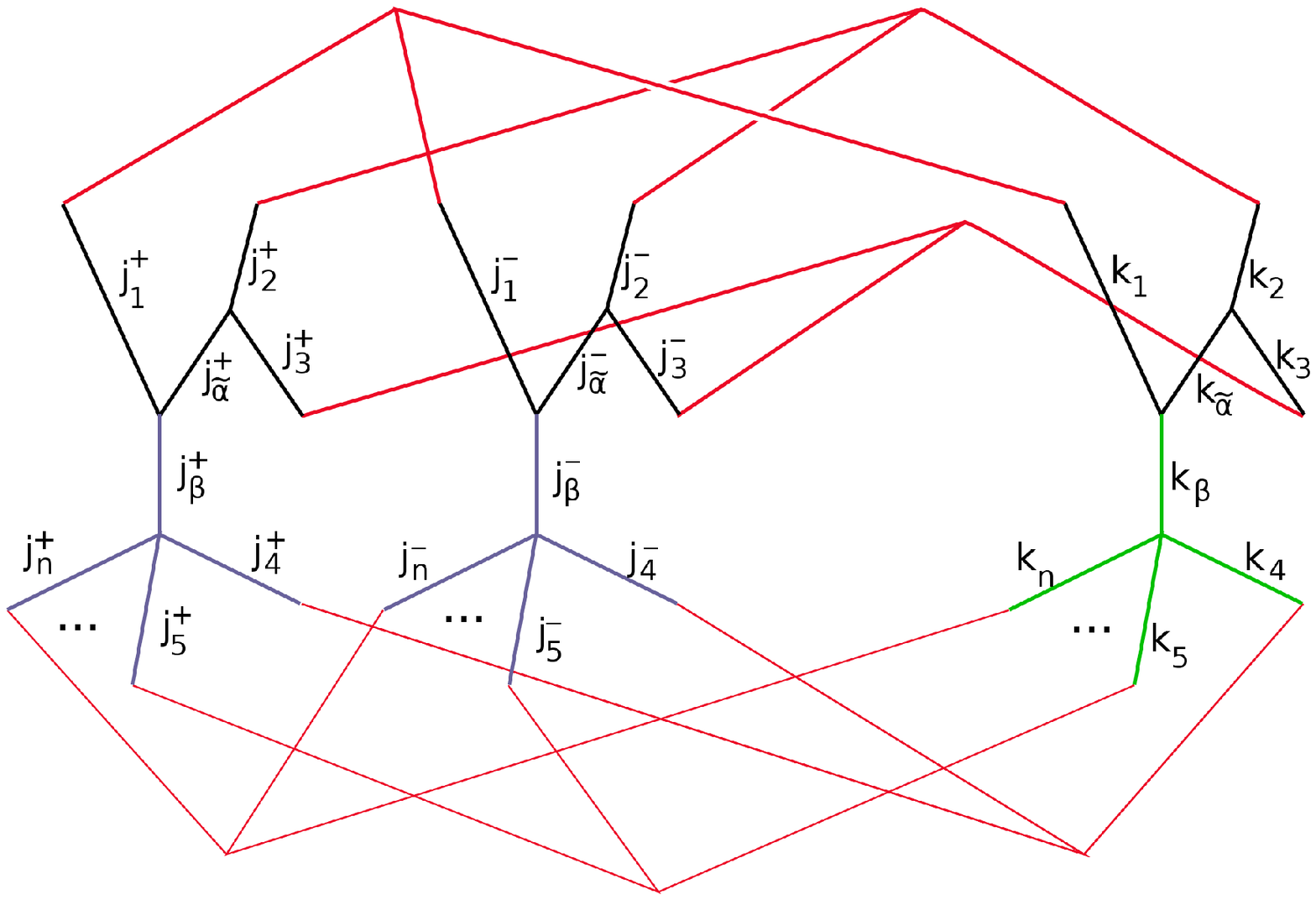}
\caption{Schematic picture of a term in the sum \eqref{level_2_empty_sum}, i.e.
$
\langle \phi, {\color{czerwony} \iota'_{k_1\ldots k_n}}({C}^{k_2 k_3}_{k_{\widetilde \alpha}} \circ C^{k_{\widetilde \alpha} k_1}_{k_\beta} \circ {\color{zielony} \omicron^{k_{\widetilde\alpha} k_{\beta}}})\rangle
 = \langle C^{j_2^+ j_3^+}_{j_{\widetilde\alpha}^+} \circ C^{j_{\widetilde\alpha}^+ j_1^+ }_{j_{\beta}^+} \circ C^{j_2^- j_3^-}_{j_{\widetilde \alpha}^-}\circ C^{j_{\widetilde\alpha}^- j_1^-}_{j_{\beta}^-}\circ {\color{niebieski} \phi^{k'_{\widetilde\alpha} k_3}}, {\color{czerwony}\iota'_{k_1\ldots k_n}}({C}^{k_2 k_3}_{k_{\widetilde \alpha}} \circ C^{k_{\widetilde \alpha} k_1}_{k_\beta} \circ {\color{zielony} \omicron^{k_{\widetilde\alpha} k_{\beta}}})\rangle
$}
\label{fig:K_empty}
\end{figure}

\begin{align*}
\langle \phi, \iota'_{k_1\ldots k_n}(C^{k_1 k_2}_{k_\alpha}\circ C^{k_\alpha k_3}_{k_\beta}\circ \omicron^{k_{\alpha} k_{\beta}})\rangle=
({\phi}^{k'_{\widetilde\alpha} k'_{\beta}})^\dagger_{A_{\beta}^+ A_4^+\ldots A_n^+, A_{\beta}^- A_4^{-}\ldots A_n^{-}} C_{A_{\widetilde \alpha}^+ A_1^+}^{A_{\beta}^+} & C_{A_2^+ A_3^+}^{A_{\widetilde\alpha}^+} C_{A_{\widetilde\alpha}^- A_1^-}^{A_{\beta}^-} C_{A_2^- A_3^-}^{A_{\widetilde\alpha}^-}\\ C_{A_1}^{A_1^+ A_1^-} \ldots C_{A_n}^{A_n^+ A_n^-} C^{A_2 A_3}_{A_{\widetilde\alpha}} {C}^{A_{\widetilde\alpha} A_1}_{A_\beta} (\omicron^{k_{\widetilde\alpha} k_{\beta}})^{A_{\beta} A_4 \ldots A_n} = {\color{blue} C_{A_2^+ A_3^+}^{A_{\widetilde\alpha}^+}  C_{A_2^- A_3^-}^{A_{\widetilde\alpha}^-} C_{A_2}^{A_2^+ A_2^-} }&{\color{blue} C_{A_3}^{A_3^+ A_3^-}  C^{A_2 A_3}_{A_{\widetilde\alpha}} } \\ {\color{red} C_{A_{\widetilde\alpha}^+ A_1^+}^{A_{\beta}^+} C_{A_{\widetilde\alpha}^- A_1^-}^{A_{\beta}^-}  C_{A_1}^{A_1^+ A_1^-}  {C}^{{A_{\widetilde\alpha}} A_1}_{A_\beta}} ({\phi}^{k'_{\widetilde\alpha} k'_{\beta}})^\dagger_{A_{\beta}^+ A_4^+\ldots A_n^+, A_{\beta}^- A_4^{-}\ldots A_n^{-}} C_{A_4}^{A_4^+ A_4^-} \ldots C_{A_n}^{A_n^+ A_n^-} & (\omicron^{k_{\widetilde\alpha} k_{\beta}})^{A_{\beta} A_4 \ldots A_n}
\end{align*}

Using lemma \ref{nonzero} we obtain, that for some $\rho_1\not=0$
\[ {\color{blue} C_{A_2^+ A_3^+}^{A_{\widetilde\alpha}^+} C_{A_2^- A_3^-}^{A_{\widetilde\alpha}^-} C_{A_2}^{A_2^+ A_2^-} C_{A_3}^{A_3^+ A_3^-} C^{A_2 A_3}_{A_{\widetilde\alpha}}} =\left\{\begin{array}{ll}
                    0, & k_{\widetilde\alpha}>j_{\widetilde\alpha}^++j^-_{\widetilde\alpha}\\
                     \rho_1 {\color{DarkOrchid} C_{A_{\widetilde\alpha}}^{A_{\widetilde\alpha}^+ A_{\widetilde\alpha}^-}}, & k_{\widetilde\alpha}=j_{\widetilde\alpha}^++j^-_{\widetilde\alpha}
                   \end{array}\right.
\] 
Applying this lemma second time we get, that for some $\rho_2\not=0$
\[ {\color{red} C_{A_{\widetilde\alpha}^+ A_1^+}^{A_{\beta}^+} C_{A_{\widetilde\alpha}^- A_1^-}^{A_{\beta}^-}}{\color{DarkOrchid} C_{A_{\widetilde\alpha}}^{A_{\widetilde\alpha}^+ A_{\widetilde\alpha}^-}} {\color{red} C_{A_1}^{A_1^+ A_1^-}  {C}^{{A_{\widetilde\alpha}} A_1}_{A_\beta}} =\left\{\begin{array}{ll}
                    0, & k_\beta>j_\beta^++j^-_\beta\\
                    \rho_2 C_{A_\beta}^{A_\beta^+ A_\beta^-}, & k_\beta=j_\beta^++j^-_\beta
                   \end{array}\right.\]
Finally, for $\rho:=\rho_1 \rho_2 \not= 0$ :
\begin{equation}\label{level_2_empty_eq}
 \langle \phi, \iota'_{k_1\ldots k_n}(C^{k_2 k_3}_{k_{\widetilde\alpha}}\circ C^{k_{\widetilde\alpha} k_1}_{k_\beta} \omicron^{k_{\widetilde\alpha} k_{\beta}})\rangle=\left\{\begin{array}{ll}
                    0, & k_{\widetilde\alpha}>k_{\widetilde\alpha}'\ {\rm or}\ k_\beta>k_\beta'\\
                    \rho\langle \phi^{k'_{\widetilde\alpha} k'_{\beta}}, \iota'_{k'_\beta \ldots k_n}(\omicron^{k'_{\widetilde\alpha} k'_{\beta}})\rangle, & k_{\widetilde\alpha}=k_{\widetilde\alpha}'\ {\rm and}\ k_\beta=k_\beta'\\
  *, & {\rm otherwise}
                   \end{array}\right.
\end{equation}
\item Now we use formula just obtained \eqref{level_2_empty_eq}  to calculate the sum \eqref{level_2_empty_sum}.

First notice that in this case $k_\beta\geq k_3$. Moreover, the elements in the sum \eqref{level_2_empty_sum} with $k_\beta>k_3$ are vanishing \eqref{level_2_empty_eq} and the sum is actually over $L'$:
\[
 \langle \phi, \iota'_{k_1\ldots k_n}(\omicron)\rangle = \sum_{(k_{\widetilde\alpha}, k_{\beta})\in L'} \langle \phi, \iota'_{k_1\ldots k_n}({C}^{k_2 k_3}_{k_{\widetilde \alpha}} \circ C^{k_1 k_{\widetilde \alpha}}_{k_\beta}\circ \omicron^{k_{\alpha} k_{\beta}})\rangle
\]
However from the definition of $k_{\widetilde \alpha}'$ follows that
\[
 \langle \phi, \iota'_{k_1\ldots k_n}(\omicron)\rangle = \sum_{k_{\widetilde\alpha}\geq k_{\widetilde\alpha}'} \langle \phi, \iota'_{k_1\ldots k_n}({C}^{k_2 k_3}_{k_{\widetilde \alpha}} \circ C^{k_1 k_{\widetilde \alpha}}_{k_3}\circ \omicron^{k_{\widetilde\alpha} k_3})\rangle
\]
Finally, using \eqref{level_2_empty_eq} we obtain:
\[
 \langle \phi, \iota'_{k_1\ldots k_n}(\omicron)\rangle = \rho \langle \phi^{k'_{\widetilde\alpha} k_3}, \iota'_{k_3 \ldots k_n}(\omicron^{k'_{\widetilde\alpha} k_3})\rangle
\]
and 
\[
 \langle \phi, \iota'_{k_1\ldots k_n}(\omicron)\rangle \not= 0.
\]
\end{enumerate}
\end{enumerate}

\vskip 0.5cm
\noindent{\bf Acknowledgement:} 
 The work was partially supported by the grants {\bf N N202 104838}, {\bf N N202 287538}, and {\bf 182/N-QGG/2008/0} (PMN) of Polish Ministerstwo Nauki i Szkolnictwa Wy\.zszego. Wojciech Kami{\'n}ski is partially supported by grant {\bf N N202 287538}. Marcin
Kisielowski acknowledges financial support from the project ”International PhD
Studies in Fundamental Problems of Quantum Gravity and Quantum Field Theory” of Foundation
for Polish Science, co-financed from the programme {\bf IE OP 2007-2013} within European Regional
Development Fund.

\end{document}